\documentclass[11pt]{article}
\usepackage[left=1in,top=1in,right=1in,bottom=1in,head=.1in,nofoot]{geometry}

\usepackage{setspace,url,bm,amsmath} 

\usepackage{titlesec} 
\titlelabel{\thetitle.\quad} 

\usepackage{graphicx} 
\usepackage{bbm}
\usepackage{latexsym}
\usepackage{caption}
\usepackage[margin=20pt]{subcaption}
\usepackage{hyperref}

\usepackage[table]{xcolor}

\newcommand{\GG}[1]{}

\usepackage{amsthm}
\usepackage{undertilde}
\theoremstyle{definition}
\newtheorem{assumption}{Assumption}
\newtheorem*{theorem*}{Theorem}
\newtheorem{theorem}{Proposition}

\newtheorem{lemma}{Lemma}
\newtheorem{example}{Example}

\newtheorem{corollary}{Corollary}
\newtheorem*{corollary*}{Corollary}

\usepackage{natbib} 
\bibpunct{(}{)}{;}{a}{}{,} 

\usepackage{etoolbox} 
\apptocmd{\sloppy}{\hbadness 10000\relax}{}{} 

\usepackage{color}
\usepackage{listings}

\def\ind{\begin{picture}(9,8)
         \put(0,0){\line(1,0){9}}
         \put(3,0){\line(0,1){8}}
         \put(6,0){\line(0,1){8}}
         \end{picture}
        }

\def\iidsim{\stackrel{i.i.d.}{\sim}}

\begin{document}
\doublespacing
\title{\bf Principal stratification analysis using principal scores}
\author{Peng Ding\footnote{University of California, Berkeley, California, U.S.A. Address for correspondence: Peng Ding, 425 Evans Hall, Berkeley, California 94720, USA.
Email: \texttt{pengdingpku@berkeley.edu}}~
and Jiannan Lu\footnote{
Microsoft Corporation, Redmond, Washington, U.S.A.
}}
\date{}
\maketitle

\begin{abstract}
Practitioners are interested in not only the average causal effect of the treatment on the outcome but also the underlying causal mechanism in the presence of an intermediate variable between the treatment and outcome. However, in many cases we cannot randomize the intermediate variable, resulting in sample selection problems even in randomized experiments. Therefore, we view randomized experiments with intermediate variables as semi-observational studies. In parallel with the analysis of observational studies, we provide a theoretical foundation for conducting objective causal inference with an intermediate variable under the principal stratification framework, with principal strata defined as the joint potential values of the intermediate variable. Our strategy constructs weighted samples based on principal scores, defined as the conditional probabilities of the latent principal strata given covariates, without access to any outcome data. This principal stratification analysis yields robust causal inference without relying on any model assumptions on the outcome distributions. We also propose approaches to conducting sensitivity analysis for violations of the ignorability and monotonicity assumptions, the very crucial but untestable identification assumptions in our theory. When the assumptions required by the classical instrumental variable analysis cannot be justified by background knowledge or cannot be made because of scientific questions of interest, our strategy serves as a useful alternative tool to deal with intermediate variables. We illustrate our methodologies by using two real data examples, and find scientifically meaningful conclusions.

\medskip
\noindent \textbf{Keywords:} Causal inference; Exclusion restriction; Intermediate variable; Monotonicity; Nonparametric identification; Principal ignorability; Sensitivity analysis.
\end{abstract}

\newpage
\section{Causal Inference with Intermediate Variables}\label{sec:intro}

When an intermediate variable between the treatment and outcome exists, often researchers are interested in not only the average causal effect of the treatment on the outcome but also the underlying causal mechanism in the presence of the intermediate variable. Naive analysis by conditioning on the observed value of the intermediate variable does not yield valid causal interpretations without imposing strong assumptions. Principal stratification \citep{Frangakis:2002}, defined as the joint potential values of the intermediate variable under both treatment and control, can be viewed as a pretreatment covariate unaffected by the treatment. Therefore, conditioning on principal stratification yields subgroup causal effects.

The subgroup causal effects classified by principal stratification have clear scientific meanings in various settings. For instance, when the intermediate variable is the actual treatment received, the principal stratification variable indicates the compliance status, and the classical instrumental variable estimator identifies the average causal effect for compliers \citep{Angrist:1996}. When the intermediate variable is the indicator for survival status, the only sensible subgroup causal effect on the outcome is the one for survivors who would potentially survive under both treatment and control \citep{Rubin:2006}. When the intermediate variable is a surrogate for the outcome, we want to predict the causal effect on the outcome by the causal effect on the surrogate. An ideal surrogate must satisfy the causal necessity that zero effect on the surrogate implies zero effect on the outcome \citep{Frangakis:2002} and the causal sufficiency that positive effect on the surrogate implies positive effect on the outcome \citep{Gilbert:2008}. Therefore, we can assess these requirements for an ideal surrogate by conducting a principal stratification analysis.

Principal stratification clarifies causal inference with intermediate variables, but it also results in inferential difficulties because of the missingness of the principal stratification variable and the consequential mixture distributions of the observed data.  We can sharpen inference about causal effects within principal strata only if we impose some of the following structural or modeling assumptions: (1) {\it monotonicity} that the treatment has a nonnegative effect on the intermediate variable for each unit \citep[e.g.,][]{Angrist:1996, Gilbert:2008}; (2) {\it exclusion restriction} that zero effect on the intermediate variable implies zero effect on the outcome \citep[e.g.,][]{Angrist:1996}; (3) {\it Normal outcome distributions} within principal strata \citep[e.g.,][]{Zhang:2009, Frumento:2012}; (4) {\it additional covariates or secondary outcomes} \citep{Ding:2011, Mattei:2011, Mattei:2013, Mealli:2013, Yang:2015, Jiang:2016}. For instance, the classical instrumental variable analysis requires exclusion restriction \citep{Angrist:1996}, which may not be justified by background knowledge or cannot be assumed according the scientific questions of interest. Without exclusion restriction, \citet{Zhang:2009} and \citet{Frumento:2012} assumed Normal outcome models within principal strata, and thus identifiability of the causal effects within principal strata is ensured by identifiability of the Normal Mixture Model. Unfortunately, the results are sensitive to the parametric modeling assumption, and the unbounded likelihood function jeopardizes statistical inference even with correctly specified model \citep{Ding:2011, Mealli:2015}. Without these assumptions, however, large sample bounds of causal effects are often too wide to be informative \citep{Zhang:2003, Cheng:2006}. We will review more applications and further highlight the inferential difficulty of principal stratification without exclusion restriction in Section \ref{sec:assum}.

Recognizing the salient feature that the intermediate is not randomized even though the treatment is randomized, we take an alternative perspective, viewing the problem as a semi-observational study. For objective causal inference, \citet{Rubin:2007, Rubin:2008} advocated designs of observational studies without access to the outcome data, which prevents data snooping by selecting favorable outcome models. In parallel with this classical wisdom of propensity scores in observational studies \citep{Rosenbaum:1983biometrika}, we propose to conduct principal stratification analysis based on principal scores, defined as the conditional probabilities of the latent principal strata given a rich set of covariates that ensure certain ignorability assumptions. Previously, applied researchers \citep{Follman:2000, Hill:2002, Jo:2009, Jo:2011, Stuart:2011} used principal scores to analyze data subject to one-sided noncompliance, and theoretical researchers \citep{Joffe:2007} suggested using principal scores to identify general causal effects within principal strata. We advance the literature by providing the theoretical foundation for using principal scores in the analysis of randomized experiments with intermediate variables. To be more specific, we give the assumptions for identification, extend previous literature to deal with general principal stratification problems beyond one-sided noncompliance, and propose statistically efficient and numerically stable weighting estimators for causal effects. The theoretical results allow for a two-step inferential procedure: we first construct weighted samples without access to the outcome data, and we then obtain simple weighting estimators for causal effects within principal strata. The whole inferential procedure does not involve any model assumptions of the outcomes, leading to more objective causal inference.

Furthermore, the central role of principal scores relies on certain ignorability and monotonicity assumptions. In parallel with sensitivity analysis in observational studies \citep{Rosenbaum:1983jrssb, Rosenbaum:2002}, we propose approaches to conducting sensitivity analysis for violations of the ignorability assumptions. Previous literature either dealt with binary outcomes  \citep[e.g.,][]{Sjolander:2009, Schwartz:2012} or relied on modeling assumptions on the outcomes \citep[e.g.,][]{Gilbert:2003}, but our strategy deals with general outcomes and relies on less modeling assumptions. Other than a few exceptions \citep{Zhang:2009, Ding:2011, Frumento:2012}, most principal stratification analyses assumed monotonicity which might be too restrictive for some applications. Our sensitivity analysis technique further removes the monotonicity assumption, and assesses the impact of its violations. The ignorability and monotonicity assumptions, though crucial for identifying the causal effects of interest, cannot be validated by observed data. Therefore, we argue that principal stratification analyses should always come with sensitivity analysis for violations of these assumptions.

The paper proceeds as follows. Section \ref{sec:assum} reviews the basic framework of principal stratification. Section \ref{sec:theory} defines principal scores and provides sufficient conditions for identifying causal effects within principal strata. Section \ref{sec::central-role-principal-scores} highlights the balancing properties of principal scores. Section \ref{sec:estimation} discusses estimation strategies that are efficient, stable and easy to implement. Section \ref{sec:sen} proposes approaches to conducting sensitivity analysis for the ignorability and monotonicity assumptions. We conduct simulation studies in Section \ref{sec:simu}, apply our methodologies to real data examples in Section \ref{sec:apply}, and conclude in Section \ref{sec::discussion}. We provide proofs and technical details in the on-line supplementary material.

%
%

\section{Potential Outcomes and Principal Stratification}\label{sec:assum}

Consider a randomized controlled experiment with $N$ units. We collect pretreatment covariates $\bm{X}_i$ for each unit $i$ before the experiment. Let $Z_i$ be the treatment assignment for unit $i$, with $Z_i=1$ for treatment and $Z_i=0$ for control. We also collect the outcome of interest $Y_i$ for unit $i$, which can be general (continuous, binary, time-to-event, etc.). In practice, we may also collect some intermediate variables between the treatment and outcome that are helpful to explain the underlying causal mechanism and treatment effect heterogeneity. We will first focus on the case with a binary intermediate variable $S$, because the binary case has the widest applications as illustrated by examples in a later part of this section. We will also comment on general $S$ later.

We use the potential outcomes framework to define causal effects. Under the Stable Unit Treatment Value Assumption (SUTVA; \citealp{Rubin:1980}), there is only one version of the treatment, and there is no interference between units. The SUTVA allows us to define the potential values of the intermediate variable and outcome for unit $i$ as $S_i(t)$ and $Y_i(t)$ under treatment $t$ for $t=0$ and $1$. Completely randomized experiments satisfy the following treatment assignment mechanism, which we will make use of throughout the paper.

\begin{assumption}
[Randomization]
\label{assume::randomization}
$Z\ind \{  S(1), S(0), Y(1), Y(0), \bm{X}  \}$.
\end{assumption}

\cite{Frangakis:2002} introduced the notion of principal stratification, defined as the joint potential values of the intermediate variable $U_i =  \{  S_i(1) , S_i(0) \} \in \{0,1\}^2$. For simplicity, we relabel the possible values of $U$, $ \{  (1,1), (1,0), (0,1),  (0,0)\} $, as $\{  ss, s\bar{s}, \bar{s}s,  \bar{s}\bar{s} \}$, respectively.
Because the principal stratification variable is unaffected by the treatment, inference conditional on $U$ yields a subgroup causal interpretation, captured by the following principal causal effect (PCE):
\begin{eqnarray*}
ACE_u = E\{  Y(1) - Y(0)\mid U=u  \}   \quad (u  =  ss, s\bar{s}, \bar{s}s, \bar{s}\bar{s}  ). 
\end{eqnarray*}
The notion of PCE is not only mathematically sound for causal evaluations, but also of scientific interest in practice.
Below, we review some important empirical applications, and discuss the scientific meanings of PCEs in each case.

\begin{example}
[Noncompliance]\label{eg::noncomplicance}
Let $S$ denote the actual treatment received, and noncompliance occurs if the treatment assignment differs from the treatment received. \cite{Angrist:1996} called $ ss, s\bar{s}, \bar{s}s,$ and $ \bar{s}\bar{s} $ always-taker, complier, defier, and never-taker, respectively. 
\end{example}

\begin{example}
[Truncation by death]\label{eg::truncation-by-death}
When some units die before measurements of their outcomes $Y$, the truncation by death problem occurs. Let $S$ be the survival status, with $S= 1$ for survival and $S=0$ for dead.
For dead patients with $S=0$, the corresponding outcome $Y$ is not well-defined.
\cite{Rubin:2006} argued that the only scientifically meaningful subgroup causal effect is $ACE_{ss}$, the survivor average causal effect, defined as the average causal effect among units who will potentially survive under both treatment and control.
Other subgroup causal effects are not well defined due to their unmeasured outcome under either treatment or control or both.
\end{example}

There are at least two problems similar to truncation by death.
In labor economics where $S$ is the employment status and $Y$ is the income, the only sensible causal effect is $ACE_{ss}$, the average causal effect among the always employed units (\citealp{Zhang:2003}, \citealp{Zhang:2009}).
In vaccine trials where $S$ is the infection status and $Y$ is a post-infection outcome, we are interested in the causal effect of vaccine on the post-infection outcome among units who would develop infection under both treatment and control \citep{Gilbert:2003, Hudgens:2006}.

\begin{example}
[Surrogate]\label{eg::surrogate-evalution}
Surrogate is of great importance in clinical trials, when the measurement of the primary outcome is costly or time-consuming.
Let $S$ denote the surrogate for the outcome $Y$. The goal of using the surrogate is to predict the causal effect on the outcome by the causal effect on the surrogate.
\cite{Frangakis:2002} argued that a good surrogate should satisfy the ``causal necessity,'' i.e., whenever the treatment has no effect on the surrogate, it has no effect on the outcome ($ACE_{ss} = 0$ and $ACE_{\bar{s}\bar{s}} = 0$). As a complement, \cite{Gilbert:2008} further argued that a good surrogate also should satisfy the ``causal sufficiency,'' i.e., whenever the treatment affects the surrogate, it also affects the outcome ($ACE_{s\bar{s}} \neq 0$ and $ACE_{\bar{s}s} \neq 0$).
\end{example}

In practice, a particular data set may simultaneously have more than one of the problems discussed in Examples \ref{eg::noncomplicance}--\ref{eg::surrogate-evalution} (\citealp{Mattei:2007}, \citealp{Frumento:2012}).
In all the examples above, estimation of $ACE_u$ is crucial for the substantive questions of interest. However, the inherent missingness of $U$, due to the ability of measuring only one of $S(1)$ and $S(0)$, jeopardizes the identification of the PCEs without some additional assumptions. In the following, we review some commonly-used assumptions, and discuss their plausibility and limitations.

\begin{assumption}
[Strong Monotonicity]\label{assume::strong-monotonicity}
$S_i(0) = 0$ for all $i$.
\end{assumption}

In Example \ref{eg::noncomplicance} of noncompliance, when the control units have no access to receive the active treatment, Strong Monotonicity holds by the design of experiments. It is sometimes referred to as one-sided noncompliance \citep{Imbens:2015}, which allows us to rule out the always takers ($U=ss$) and defiers ($U=\bar{s}s$).
In the literature on surrogates, Strong Monotonicity is closely related to the ``constant biomarker'' assumption (\citealp{Gilbert:2008}).

\begin{assumption}
[Monotonicity]\label{assume::monotonicity}
$S_i(1) \geq S_i(0)$ for all $i$.
\end{assumption}

Monotonicity rules out only the defiers ($U = \bar{s}s$). In general, we cannot test Monotonicity by the observed data unless $\Pr(S=1\mid Z=1) < \Pr(S=1\mid Z=0).$

\begin{assumption}
[Exclusion Restriction]\label{assume::exclusion-restriction}
$Y_i(1)=Y_i(0)$ for $U_i = ss$ and $\bar{s}\bar{s}$.
\end{assumption}

Exclusion Restriction (ER) implies that $ACE_{ss} = ACE_{\bar{s}\bar{s}}  = 0$. In Example \ref{eg::noncomplicance} of noncompliance, ER is plausible in double-blinded trials because the outcome may be affected only by the treatment received.
\cite{Angrist:1996} showed that under Monotonicity
and ER, the complier average causal effect, $ACE_{s\bar{s}}$, is identified by the ratio of the average causal effects on $Y$ and $S.$

However, in many open-label trials, the treatment assignment may have a ``direct effect'' on the outcome, and ER may not hold \citep[e.g.,][]{Hirano:2000}. What is more important, we cannot assume ER in the truncation by death and surrogate problems, because in these settings it is the question of concern to test whether $ACE_{ss}$ or $ACE_{\bar{s}\bar{s}}$ is zero. Imposing ER immediately discards the very scientific question of interest, which is not reasonable. Unfortunately, if we do not impose ER due to either background knowledge or substantive questions of interest, we can no longer nonparametrically identify the PCEs without further assumptions. In this paper, we will discuss alternative sufficient conditions that ensure nonparametric identifiability of the PCEs, and propose estimators that rely on minimal modeling assumptions.

\section{Nonparametric Identification of Principal Causal Effects}\label{sec:theory}

Our identification strategy, in parallel with the notion of propensity score in observational studies, exploits principal scores defined as the probabilities of the latent principal strata given a rich set of pretreatment covariates. Although in the existing literature the principal scores are used in the one-sided noncompliance problem, its rigorous theoretical foundation is lacking, and more importantly it cannot deal with more general cases. We fill in the gap by demonstrating some general identification results based on  principal scores.

\subsection{Principal Scores}

Although we cannot uniquely recover the unobserved principal strata indicators, we can create weighted samples based on principal scores:
$$
e_u(\bm{X}) = \Pr(U =  u \mid \bm{X}) \quad (u =  ss,s\bar{s},\bar{s}s,\bar{s}\bar{s}  ).
$$
Define $p_1=\Pr(S=1\mid Z=1)$ and $p_0=\Pr(S=1\mid Z=0)$ as the probabilities of $S$ under treatment and control, and the analogous conditional probabilities as $p_1(\bm{X})=\Pr(S=1\mid Z=1, \bm{X})$ and $p_0(\bm{X})=\Pr(S=1\mid Z=0, \bm{X})$ given covariates $\bm{X}$.

Under Strong Monotonicity, two strata $s\bar{s}$ and $\bar{s}\bar{s}$ exist. The observed data with $(Z=1,S=1)$ contain only strata $U=s\bar{s}$, and the observed data with $(Z=1,S=0)$ contain only strata $U=\bar{s}\bar{s}.$
Therefore, we can use the treatment arm to identify the principal scores by
$
e_{s\bar{s}}(\bm{X})   = p_1(\bm{X})
$
and
$
e_{\bar{s}\bar{s}}(\bm{X})   = 1- p_1(\bm{X}) ,
$
and the proportions of the two principal strata by $\pi_{s\bar{s}} = p_1$ and $\pi_{\bar{s}\bar{s}} = 1-p_1$.

Under Monotonicity, three strata $ss$, $s\bar{s}$, and $\bar{s}\bar{s}$ exist. The observed data with $(Z=1,S=0)$ contain only strata $U=\bar{s}\bar{s}$, and the observed data with $(Z=0,S=1)$ contain only strata $U=ss$. Therefore, we can identify the principal scores by
$
e_{ss}(\bm{X}) = p_0(\bm{X}),
$
$
e_{\bar{s}\bar{s}} (\bm{X}) = 1-p_1(\bm{X}),
$
and
$
e_{s\bar{s}}(\bm{X}) = p_1(\bm{X})-p_0(\bm{X}),
$
and the proportions of the three principal strata by $\pi_{ss} =p_0, \pi_{\bar{s}\bar{s}} =1-p_1$, and $\pi_{s\bar{s}} = p_1-p_0$.

The above discussion demonstrates nonparametric identification of principal scores under (Strong) Monotonicity. We postpone the discussion of  modeling principal scores to Section \ref{subsec::ps}.

\subsection{Principal Ignorability and Nonparametric Identification}

The observed data are mixtures of at most two latent principal strata. Our goal is to disentangle the latent components of the outcome distributions. Although we can view them as the weights for the latent subgroup indicators, principal scores themselves alone are not sufficient to identify the PCEs. The following principal ignorability assumptions \citep{Jo:2009, Jo:2011, Stuart:2011}, in parallel with the ignorability assumption in observational studies \citep{Rosenbaum:1983biometrika}, are sufficient conditions to nonparametrically identify the PCEs.

\paragraph{Under Strong Monotonicity}
We invoke the following version of Principal Ignorability (PI).

\begin{assumption}
[PI]\label{assume::principal-ignorability}
$Y(0) \ind U \mid \bm{X}$.
\end{assumption}

A weaker version of the above assumption, shown in Section \ref{subsec:GPIsen}, also suffices for our later discussion on identification. Here we use a stronger version for easy interpretation. PI assumes conditional independence of $Y(0)$ and $U$ given $\bm{X}$, i.e., a random allocation of the principal stratification variable with respect to the control potential outcome given $\bm{X}$. PI requires an adequate set of covariates $\bm{X}$, conditional on which there is no difference between the distributions of the control potential outcomes across principal strata $U=s\bar{s}$ and $U=\bar{s}\bar{s}.$

With the identifiability of the principal scores, PI further helps identify $ACE_u$.

\begin{theorem}
\label{thm::identifiability-strong-monotonicity}
Under Strong Monotonicity and PI, we can identify the PCEs by
\begin{eqnarray*}
ACE_{s\bar{s}} &=& E(Y\mid Z=1, S=1) - E\{   w_{s\bar{s}} (\bm{X}) Y \mid Z=0   \},\\
ACE_{\bar{s}\bar{s}} &=& E(Y\mid Z=1, S=0) - E\{  w_{\bar{s}\bar{s}}(\bm{X}) Y\mid Z=0  \},
\end{eqnarray*}
where $ w_{s\bar{s}}(\bm{X}) = e_{s\bar{s}}(\bm{X}) / \pi_{s\bar{s}} $ and $ w_{\bar{s}\bar{s}} (\bm{X})  =  e_{\bar{s}\bar{s}}(\bm{X})  / \pi_{\bar{s} \bar{s}}$.
\end{theorem}

The treatment group does not involve mixture distributions.
The control group is a mixture of two strata $s\bar{s}$ and $\bar{s}\bar{s}$, and Proposition \ref{thm::identifiability-strong-monotonicity} shows that the weight $w_u(\bm{X})$ is the ratio of the principal score over the marginal proportion of stratum $u.$

\paragraph{Under Monotonicity}
We invoke the following General Principal Ignorability (GPI).

\begin{assumption}
[GPI]\label{assume::GPI}
$Y(z) \ind U \mid  \bm{X}$ for $z=0$ and $1$.
\end{assumption}

Again, a weaker version of GPI suffices to identify PCEs as discussed in Section \ref{subsec:GPIsen}, but the stronger version enjoys easier interpretation. The mathematical form of the above assumption is similar to the ignorability assumption in observational studies \citep{Rosenbaum:1983biometrika}, with $U$ being the latent principal stratification variable instead of the treatment indicator. Intuitively, the conditional independence of GPI requires enough covariates $\bm{X}$ remove all ``confounding'' between $U$ and $Y$. More precisely, conditional on $\bm{X}$, there is no difference between the distributions of the treatment potential outcomes across strata $U=ss$ and $U=s\bar{s}$, and no difference between the distributions of the control potential outcomes across strata $U=\bar{s}\bar{s}$ and $U=s\bar{s}$. These interpretations will become more apparent in Section \ref{subsec:GPIsen}. See \citet{Guo:2014} for a slightly different view on GPI.

\begin{theorem}
\label{thm::identifiability-monotonicity-GPI}
Under Monotonicity and GPI, we can identify the PCEs by
\begin{eqnarray*}
ACE_{s\bar{s}}  &=& E\{   w_{1 , s\bar{s} }(\bm{X})  Y\mid Z=1, S=1   \}  - E\{   w_{ 0 , s\bar{s}}(\bm{X} ) Y\mid Z=0,S=0  \} ,\\
ACE_{\bar{s}\bar{s}} &=& E(Y\mid Z=1,S=0) - E\{  w_{0, \bar{s}\bar{s} }(\bm{X}) Y\mid Z=0, S=0   \},\\
ACE_{ss} &=& E\{  w_{1, ss } (\bm{X})  Y\mid Z=1, S=1  \} - E(Y\mid Z=0, S=1),
\end{eqnarray*}
where
\begin{eqnarray*}
w_{1 , s\bar{s} } (\bm{X}) = \frac{ e_{s\bar{s}} (\bm{X}) }{  e_{s\bar{s}} (\bm{X}) + e_{ss} (\bm{X})  } \Big/ \frac{\pi_{s\bar{s}}}{  \pi_{s\bar{s}}   +  \pi_{ss}  } , & &
w_{0, s\bar{s} } (\bm{X}) = \frac{ e_{s\bar{s}} (\bm{X}) }{  e_{s\bar{s}} (\bm{X}) + e_{ \bar{s}\bar{s}} (\bm{X})  } \Big/ \frac{\pi_{s\bar{s}}}{  \pi_{s\bar{s}}   +  \pi_{\bar{s}\bar{s}}  } ,\\
w_{0, \bar{s}\bar{s} } (\bm{X} )  = \frac{ e_{\bar{s}\bar{s}} (\bm{X}) }{  e_{s\bar{s}} (\bm{X}) + e_{ \bar{s}\bar{s}} (\bm{X})  }   \Big/ \frac{\pi_{\bar{s}\bar{s}}}{  \pi_{s\bar{s}}   +  \pi_{\bar{s}\bar{s}}  } , &&
w_{1, ss} (\bm{X}) = \frac{ e_{ss} (\bm{X}) }{  e_{s\bar{s}} (\bm{X}) + e_{ss} (\bm{X})  } \Big/ \frac{\pi_{ss}}{  \pi_{s\bar{s}}   +  \pi_{ss}  } .
\end{eqnarray*}
\end{theorem}

The observed data with $(Z=1,S=0)$ and $(Z=0,S=1)$ do not involve mixture distributions.
The observed data with $(Z =1,S =1)$ contain a mixture of two strata $s\bar{s}$ and $ss$, and the weight $w_{1,u}(\bm{X})$ is the probability of stratum $u$ conditional on $(Z=1,S=1,\bm{X})$ divided by the probability conditional only on $(Z=1,S=1)$.
Similar discussion applies to the observed data with $(Z =0, S =0)$ and the weight $w_{0,u}(\bm{X})$.

\section{Balancing Properties of Principal Scores}
\label{sec::central-role-principal-scores}

\subsection{Balancing Properties}\label{subsec::psbalance}

Principal scores play a crucial role in the theory developed in the last section. Therefore, it is of practical importance to select a principal score model that is close to the truth.
Fortunately, we can use the following balancing conditions for any function of the covariates, $h(\bm{X})$, to guide our choice of the model for $\Pr(U\mid \bm{X})$.

\begin{corollary}
\label{corollary::SM}
Under Strong Monotonicity, we have
\begin{eqnarray*}
E\{ h(\bm{X})\mid Z=1, S=1\} &=& E\{  w_{s\bar{s}}(\bm{X}) h(\bm{X}) \mid Z=0\}, \\
E\{ h(\bm{X})\mid Z=1, S=0\} &=& E\{  w_{\bar{s}\bar{s}}(\bm{X})  h(\bm{X}) \mid Z=0\}.
\end{eqnarray*}
\end{corollary}

\begin{corollary}
\label{corollary::M}
Under Monotonicity, we have
\begin{eqnarray*}
E\{   w_{1 , s\bar{s} }   (\bm{X})  h(\bm{X})  \mid Z=1, S=1   \}
&=&  E\{   w_{ 0 , s\bar{s}}  (\bm{X} ) h(\bm{X}) \mid Z=0,S=0  \} ,\\
E\{ h(\bm{X})  \mid Z=1,S=0 \}
&=& E\{  w_{0, \bar{s}\bar{s} } (\bm{X}) h(\bm{X}) \mid Z=0, S=0   \},\\
E\{  w_{1, ss }   (\bm{X})  h(\bm{X})   \mid Z=1, S=1  \}
&=& E\{ h(\bm{X})\mid Z=0, S=1\}.
\end{eqnarray*}
\end{corollary}

The above corollaries are direct applications of Propositions \ref{thm::identifiability-strong-monotonicity} and \ref{thm::identifiability-monotonicity-GPI}. Intuitively, because any functions of the covariates $h(\bm X)$ are unaffected by the treatment within principal strata, the ``PCEs'' on $h(\bm X)$ are all zeros.
Although simple, the balancing conditions in Corollaries \ref{corollary::SM} and \ref{corollary::M} allow for model checking for principal scores, and are therefore of practical importance. If the balancing conditions above are obviously violated, we need to build a more flexible model to account for the residual dependence of $U$ on $\bm{X}$.
For example, we can add higher order polynomial and interaction terms of the covariates into the Logistic model, until the balancing conditions are well satisfied.
This idea is similar to designs of observational studies for achieving objective causal inference \citep{Rubin:2007, Rubin:2008}. When constructing weighted samples, we do not have access to the outcome data, because we require only $(\bm{X}, Z, S)$ for creating the principal score estimates. This outcome-free strategy for designs, advocated by \citet{Rubin:2007, Rubin:2008} and \citet{Imbens:2015}, has the merit of being free of data snooping based on repeated search for favorable outcome models.

\subsection{Estimating Principal Scores}\label{subsec::ps}
Although we have nonparametric identification results under the PI assumptions \ref{assume::principal-ignorability} and \ref{assume::GPI}, we can easily deal with only low dimensional and discrete covariates to estimate the principal scores. With high dimensional or continuous covariates, we need to specify models for $\Pr(U\mid\bm{X})$.

Under Strong Monotonicity, $U$ takes only two values, and we can use a Logistic model for $ \Pr( U\mid\bm{X} ) $.
By Randomization, we can fit a Logistic model of $S$ on $\bm{X}$ using only the data from the treatment group, because within arm $Z=1$, we have $S=1$ if and only if $U=s\bar{s}$, and $S=0$ if and only if $U=\bar{s}\bar{s}$.

Under Monotonicity, $U$ takes three values, we can model $\Pr(U\mid\bm{X})$ as a three-level Multinomial Logistic model,
and use the EM algorithm \citep{dempster1977maximum} to find the Maximum Likelihood Estimates (MLEs) by treating $U$ as missing data. See the supplementary material for computational details.

In practice, correct specification of the principal score model $\Pr(U\mid \bm{X})$ is crucial for the validity of the principal causal effect estimation, because misspecification of $\Pr(U\mid \bm{X})$ may lead to biased estimators for the PCEs. After fitting a principal score model, we can use Corollaries \ref{corollary::SM} and \ref{corollary::M} to check balance of some important covariates and their functions. If the balancing conditions are violated, we can fit a more flexible model (e.g., adding high order polynomials or interaction terms of $\bm{X}$ into the Logistic models) until the balance conditions are satisfied.

\section{Modeling the Outcome and Model-Assisted Estimators}\label{sec:estimation}
Previous identification and sensitivity analysis results assume infinite amounts of data or a known distribution of the observed data. In this section, we discuss finite sample estimators of PCEs. For simplicity, in the main text we will discuss only the estimator for $ACE_{s\bar{s}}$ under Monotonicity. We have similar results for other strata, the cases under Strong Monotonicity, and the cases for sensitivity analysis; we relegate the technical details to the supplementary material.

The identification formulas in Propositions \ref{thm::identifiability-strong-monotonicity}--\ref{thm::PCE-without-monotonicity} immediately give us simple moment estimators by weighting, with $e_u(\bm{X})$ and $\pi_u$ replaced by their consistent estimators, and the expectations over the population replaced by their sample analogues.
In the above discussion about identification and moment estimators for PCEs, we use the covariates to predict latent strata and create weights. In fact, covariates contain useful information about both the principal strata and the outcome distributions. Now we will use covariate adjustment to improve statistical efficiency for estimation.
Covariate adjustment is based on the following simple fact that for all $u$ and all fixed vectors $\bm{\beta}_{z,u}$,
\begin{eqnarray}
ACE_u = E\{  Y(1)  - \bm{\beta}_{1,u}^\top \bm{X} \mid U=u \} - E\{  Y(0) - \bm{\beta}_{0, u}^\top \bm{X} \mid U=u \} + (\bm{\beta}_{1,u} - \bm{\beta}_{0,u})^\top E(\bm{X}\mid U=u) .
\label{eq::identity}
\end{eqnarray}
Treating the ``residual'' $Y(z)  - \bm{\beta}_{z,u}^\top \bm{X}$ as a new ``potential outcome,'' we can apply Proposition \ref{thm::identifiability-monotonicity-GPI} to identify three expectation terms in formula (\ref{eq::identity}) via the following corollary.

\begin{corollary}
\label{coro::covariance-adjustment-mono-complier}
Under Monotonicity and GPI, we have
\begin{eqnarray*}
E\{  Y(1)  - \bm{\beta}_{1,s\bar{s}}^\top \bm{X} \mid U=s\bar{s} \} &=&
E\{  w_{1 , s\bar{s} }(\bm{X})(Y  - \bm{\beta}_{1, s\bar{s}}^\top \bm{X}) \mid Z=1, S=1 \}, \\
E\{  Y(0) - \bm{\beta}_{0, s\bar{s}}^\top \bm{X} \mid U=s\bar{s} \}  &=&
E\{   w_{ 0 , s\bar{s}}(\bm{X} ) ( Y  - \bm{\beta}_{0, s\bar{s}}^\top \bm{X} )  \mid  Z=0, S=0  \}, \\
E(\bm{X}\mid U=s\bar{s})  & = & E\{ w_{1 , s\bar{s} }(\bm{X})\bm{X}\mid Z=1,S=1\} \; = \; E\{ w_{0 , s\bar{s} }(\bm{X})  \bm{X}  \mid Z=0, S=0 \}.
\end{eqnarray*}
\end{corollary}

Define $n_{zs} = \#\{i: Z_i = z, S_i = s\}$.
The covariate-adjusted estimator for $ACE_{s\bar{s}}$ is
\begin{eqnarray}
\widehat{ACE}_{s\bar{s}}^{\textrm{adj}}  =
\frac{1}{n_{11}} \sum_{  \{ i: Z_i=1, S_i=1\}  } \widehat{w}_{1, s\bar{s} }  (\bm{X}_i) (Y_i   - \bm{\beta}_{1,s\bar{s}}^\top \bm{X}_i  )
 -  \frac{1}{n_{00}} \sum_{  \{i:Z_i=0, S_i=0\} }  \widehat{w}_{0, s\bar{s} }(\bm{X}_i)    (Y_i - \bm{\beta}_{0, s\bar{s}}^\top \bm{X} )
 \nonumber   \\
 +  \frac{1}{ n_{11} + n_{00}  } (\bm{\beta}_{1,s\bar{s}} - \bm{\beta}_{0, s\bar{s}} )^\top   \left\{     \sum_{  \{ i: Z_i=1, S_i=1\}  }   \widehat{w}_{1, s\bar{s} }(\bm{X}_i)\bm{X}_i +  \sum_{  \{i:Z_i=0, S_i=0\} } \widehat{w}_{0, s\bar{s} }(\bm{X}_i)  \bm{X}_i \right\}.  
\label{eq::covariate-adj-mono-complier}
\end{eqnarray}

As long as the potential outcomes are correlated with the covariates, the ``residual potential outcomes,'' $Y(z)  - \bm{\beta}_{z,u}^\top \bm{X}$, will have smaller variances than the original potential outcomes. Therefore, the covariate-adjusted estimator in formula (\ref{eq::covariate-adj-mono-complier}) tends to have a smaller asymptotic variance than the unadjusted estimator. Our simulation studies have verified this intuition. Although Corollary \ref{coro::covariance-adjustment-mono-complier} and the covariate-adjusted estimator in formula (\ref{eq::covariate-adj-mono-complier}) hold for any fixed vectors $\bm{\beta}_{1,s\bar{s}}$ and $\bm{\beta}_{0,s\bar{s}}$, we need to choose them in practice.
Intuitively, we can choose $\bm{\beta}_{z,s\bar{s}}$ as the linear regression coefficient of $Y(z)$ onto the space spanned by $\bm{X}$ for units $U=s\bar{s}$, i.e.,
$$
\bm{\beta}_{z,s\bar{s}} =   \{   E( \bm{X} \bm{X}^\top   \mid U=s\bar{s} )  \}^{-1}    E\{   \bm{X} Y(z) \mid U=s\bar{s}  \}.
$$
Similar to Corollary \ref{coro::covariance-adjustment-mono-complier}, each component of the above least squares formula is also identifiable.

\begin{corollary}
\label{coro::covariance-adjustment-mono-complier-OLS}
Under Monotonicity and GPI, we have
\begin{eqnarray*}
E\{ \bm{X} Y(1)   \mid U=s\bar{s} \} &=&
E\{  w_{1 , s\bar{s} }(\bm{X}) \bm{X} Y \mid Z=1, S=1 \}, \\
E\{  \bm{X} Y(0)   \mid U=s\bar{s} \}  &=&
E\{   w_{ 0 , s\bar{s}}(\bm{X} ) \bm{X}  Y    \mid  Z=0, S=0  \}, \\
E(\bm{X} \bm{X}^\top  \mid U=s\bar{s})  & = & E\{ w_{1 , s\bar{s} }(\bm{X})\bm{X} \bm{X}^\top \mid Z=1,S=1\} \; = \; E\{ w_{0 , s\bar{s} }(\bm{X})  \bm{X} \bm{X}^\top  \mid Z=0, S=0 \}.
\end{eqnarray*}
\end{corollary}

Therefore, we choose $\bm{\beta}_{1,s\bar{s}}$ as the weighted least squares regression coefficient of $Y_i$ on $\bm{X}_i$ using samples with $(Z_i=1,S_i=1)$ and weights $ w_{1 , s\bar{s} }(\bm{X}_i)$, and $\bm{\beta}_{0,s\bar{s}}$ as the weighted least squares regression coefficient of $Y_i$ on $\bm{X}_i$ using samples with $(Z_i=0,S_i=0)$ and weights $ w_{0 , s\bar{s} }(\bm{X}_i)$.

However, we do not assume that the response surface of $Y(z)$ on $\bm{X}_i$ is linear, and the consistency of the estimators does not rely on any modeling assumptions about $Y(z)$. 
Our estimators are essentially moment estimators, and their consistency and asymptotic Normality follow directly from standard arguments of the Law of Large Numbers and Central Limit Theorem. Therefore, they have superior statistical properties compared to principal stratification analysis based on Normal mixture models, which have unbounded likelihood and inaccurate asymptotic Normal approximations as pointed out by \citet{Mealli:2015}. Furthermore, in our simulation studies shown in the supplementary material, we compare our method and \citet{Jo:2009}'s method involving outcome modeling, and find that our estimator is not only robust to misspecification of the outcome model but also has smaller standard error.

\section{Sensitivity Analysis}\label{sec:sen}

The theoretical foundation of the identification and estimation relies crucially on Monotonicity and PI, which are fundamentally untestable. In some cases, however, these two assumptions may not be easily justified according to background knowledge. In this section, we propose approaches to conducting sensitivity analysis to assess the impact of violations of Monotonicity or PI.

\subsection{Sensitivity Analysis for Principal Ignorability}\label{subsec:GPIsen}

The PI assumptions are critical for nonparametric identification of the PCEs as shown in Section \ref{sec:theory}. They require the observed covariates $\bm{X}$ capture the key characteristics that affect both the principal stratum and the potential outcomes.
They are sufficient conditions to ensure nonparametric identification, which are similar to the ignorability assumption used in causal inference with observational studies \citep{Rosenbaum:1983biometrika} and the sequential ignorability assumption used in mediation analysis \citep[cf.][]{vanderweele2015explanation}. In many cases, the more covariates we observe, the more plausible these assumptions will become. In practice, however, we may not able to collect adequate covariates to remove the ``confounding'' between the principal stratification and the outcome variable.  Unfortunately, the PI assumptions cannot be validate by the observed data.  
Although there is a long history of sensitivity analysis in observational studies (e.g., \citealp{Rosenbaum:1983jrssb}, \citealp{Rosenbaum:2002}), there are only a few sensitivity analysis techniques for principal stratification analysis with binary outcomes \citep[e.g.,][]{Sjolander:2009, Schwartz:2012} and some modeling assumptions \citep[e.g.,][]{Gilbert:2003} under Monotonicity. We provide a more general framework to assess the sensitivity of the deviations from the PI assumptions.

\paragraph{Under Strong Monotonicity}
According to its proof in the supplementary material, Proposition \ref{thm::identifiability-strong-monotonicity} holds under a weaker version of PI, $E\{  Y(0)\mid U=s\bar{s}, \bm{X} \}  = E\{ Y(0) \mid U = \bar{s}\bar{s} , \bm{X} \}$, which requires the means of the control potential outcomes be the same for strata $U=s\bar{s}$ and $U=\bar{s}\bar{s}$ conditional on covariates $\bm{X}$. Therefore, our sensitivity analysis is based on the deviation from this weaker assumption, captured by a single sensitivity parameter
$$
\varepsilon =  \frac{ E\{  Y(0)\mid U=s\bar{s}, \bm{X} \}  }{  E\{ Y(0) \mid U = \bar{s}\bar{s} , \bm{X} \} },
$$
where we implicitly assume that $\varepsilon$ does not depend on the covariates $\bm{X}$. When the outcome is binary, the sensitivity parameter $\varepsilon$ becomes the relative risk of $U$ on the control potential outcome $Y(0)$ given covariates $\bm{X}$. When $\varepsilon=1$, the same identification results hold as those under PI. When $\varepsilon\neq 1$, we can identify the PCEs for a fixed value of $\varepsilon$, as shown in the following theorem.

\begin{theorem}
\label{thm::sensitivityanalysis-strong-monotonicity}
Under Strong Monotonicity, for a fixed value of $\varepsilon$, we can identify the PCEs by
\begin{eqnarray*}
ACE_{s\bar{s}}  &=& E(Y\mid Z=1, S=1) - E\{   w_{s\bar{s}}^\varepsilon (\bm{X}) Y \mid Z=0   \},\\
ACE_{\bar{s}\bar{s}}  &=& E(Y\mid Z=1, S=0) - E\{  w_{\bar{s}\bar{s}} ^\varepsilon (\bm{X}) Y\mid Z=0  \},
\end{eqnarray*}
where
$$
w_{s\bar{s}}^\varepsilon (\bm{X}) = \frac{  \varepsilon e_{s\bar{s}}(\bm{X}) } {   \{ \varepsilon  e_{s\bar{s}}(\bm{X}) + e_{\bar{s}\bar{s}} (\bm{X})  \} \pi_{s\bar{s}} } , \quad
 w_{\bar{s}\bar{s}}^\varepsilon (\bm{X})  =  \frac{   e_{\bar{s}\bar{s}}(\bm{X}) } {   \{ \varepsilon  e_{s\bar{s}}(\bm{X}) + e_{\bar{s}\bar{s}} (\bm{X})  \} \pi_{\bar{s}\bar{s}} } .
 $$
\end{theorem}

Although the principal scores remain the same as in Proposition \ref{thm::identifiability-strong-monotonicity}, the new weight $w_u^\varepsilon(\bm{X})$ further depends on the deviation from PI, with the principal score $e_{s\bar{s}}(\bm{X})$ over-weighted by the sensitivity parameter $\varepsilon$.

\paragraph{Under Monotonicity}
According to its proof in the supplementary material, Proposition \ref{thm::identifiability-monotonicity-GPI} holds under a weaker version of GPI, i.e.,
$
E\{  Y(1)\mid U=s\bar{s}, \bm{X} \}  = E\{ Y(1) \mid U =  ss , \bm{X} \}
$
and
$
E\{  Y(0)\mid U=s\bar{s}, \bm{X} \}  = E\{ Y(0) \mid U = \bar{s}\bar{s} , \bm{X} \},
$
which require the conditional means of $Y(1)$ be the same for strata $U=s\bar{s}$ and $U=ss$, and the conditional means of $Y(0)$ be the same for strata $U=s\bar{s}$ and $U=\bar{s}\bar{s}$ given covariates $\bm{X}$. Therefore, our sensitivity analysis is based on the deviations from the above weaker assumption, captured by the following two sensitivity parameters:
$$
\varepsilon_1 =  \frac{ E\{  Y(1)\mid U=s\bar{s}, \bm{X} \}  }{  E\{ Y(1) \mid U =  ss , \bm{X} \} },
\quad
\varepsilon_0 =  \frac{ E\{  Y(0)\mid U=s\bar{s}, \bm{X} \}  }{  E\{ Y(0) \mid U = \bar{s}\bar{s} , \bm{X} \} },
$$
where $\varepsilon_1 $ and $\varepsilon_0$ are for the potential outcomes under treatment and control, respectively.
The sensitivity parameters $\varepsilon_1 $ and $\varepsilon_0$ enjoy transparent interpretations, which allows us to select the range of them according to background knowledge. For example, in the flu shot encouragement design with noncompliance discussed in \citet{Hirano:2000}, it may be reasonable to believe that on average the never-takers ($U=\bar{s}\bar{s}$) are the strongest patients and the always-takers ($U=ss$) are the weakest patients. In this case, the outcome of interest is an indicator of flu related hospital visit, and therefore we can select sensitivity parameters within the range $\varepsilon_1<1$ and $\varepsilon_0>1.$ We will analyze this example in detail in Section \ref{subsec::flu}.

For fixed values of the sensitivity parameters $(\varepsilon_1, \varepsilon_0)$, we have the following theorem.

\begin{theorem}
\label{thm::sensitivityanalysis-monotonicity}
Under Monotonicity, and for fixed values of $(\varepsilon_1, \varepsilon_0)$, we can identify the PCEs by
\begin{eqnarray*}
ACE_{s\bar{s}}  &=&
E\{   w_{1 , s\bar{s} } ^{\varepsilon_1} (\bm{X})  Y\mid Z=1, S=1   \}
- E\{   w_{ 0 , s\bar{s}} ^{\varepsilon_0} (\bm{X} ) Y\mid Z=0,S=0  \} ,\\
ACE_{\bar{s}\bar{s}} &=&
E(Y\mid Z=1,S=0)
- E\{  w_{0, \bar{s}\bar{s} } ^{\varepsilon_0} (\bm{X}) Y\mid Z=0, S=0   \},\\
ACE_{ss} &=&
E\{  w_{1, ss } ^{\varepsilon_1} (\bm{X})  Y\mid Z=1, S=1  \}
- E(Y\mid Z=0, S=1),
\end{eqnarray*}
where
\begin{eqnarray*}
w_{1 , s\bar{s} }^{\varepsilon_1} (\bm{X})
= \frac{ \varepsilon_1 e_{s\bar{s}} (\bm{X}) }{  \varepsilon_1 e_{s\bar{s}} (\bm{X}) + e_{ss} (\bm{X})  } \Big/ \frac{\pi_{s\bar{s}}}{  \pi_{s\bar{s}}   +  \pi_{ss}  } , & &
w_{0, s\bar{s} }^{\varepsilon_0} (\bm{X})
= \frac{ \varepsilon_0 e_{s\bar{s}} (\bm{X}) }{  \varepsilon_0 e_{s\bar{s}} (\bm{X}) + e_{ \bar{s}\bar{s}} (\bm{X})  } \Big/ \frac{\pi_{s\bar{s}}}{  \pi_{s\bar{s}}   +  \pi_{\bar{s}\bar{s}}  } ,\\
w_{0, \bar{s}\bar{s} }^{\varepsilon_0} (\bm{X} )
= \frac{ e_{\bar{s}\bar{s}} (\bm{X}) }{  \varepsilon_0 e_{s\bar{s}} (\bm{X}) + e_{ \bar{s}\bar{s}} (\bm{X})  }   \Big/ \frac{\pi_{\bar{s}\bar{s}}}{  \pi_{s\bar{s}}   +  \pi_{\bar{s}\bar{s}}  } , &&
w_{1, ss}^{\varepsilon_1} (\bm{X})
= \frac{ e_{ss} (\bm{X}) }{ \varepsilon_1  e_{s\bar{s}} (\bm{X}) + e_{ss} (\bm{X})  } \Big/ \frac{\pi_{ss}}{  \pi_{s\bar{s}}   +  \pi_{ss}  } .
\end{eqnarray*}
\end{theorem}

The weights have similar adjustments as in Proposition \ref{thm::sensitivityanalysis-strong-monotonicity}, over-weighting the principal score $w_{s\bar{s}}(\bm{X})$ by the sensitivity parameters $\varepsilon_1$ and $\varepsilon_0$ in the treatment group and control group, respectively. $ACE_{s\bar{s}}$ depends on both $\varepsilon_1$ and $\varepsilon_0$, $ACE_{\bar{s}\bar{s}}$  only on $\varepsilon_0$, and $ACE_{ss}$ only on $\varepsilon_1$.

As a side note, Proposition \ref{thm::sensitivityanalysis-monotonicity} not only allows for sensitivity analysis of possible violations of GPI, but also allows for testing the fundamental assumptions of GPI and ER. To be more specific, if $ACE_{\bar{s}\bar{s}} = 0$ and $\varepsilon_0=1$, then 
\begin{eqnarray}
\label{eq::testable}
E(Y\mid Z=1,S=0)
= E\{  w_{0, \bar{s}\bar{s} }   (\bm{X}) Y\mid Z=0, S=0   \}.
\end{eqnarray}
The contrapositive states that if we reject \eqref{eq::testable} by the observed data, then we must reject $ACE_{\bar{s}\bar{s}} =  0$ or $\varepsilon_0 = 1$. Therefore, if we assume $ACE_{\bar{s}\bar{s}} =  0$, then we can test $\varepsilon_0 =  1$; if we assume $\varepsilon_0 =  1$, then we can test $ACE_{\bar{s}\bar{s}} =  0$. Analogous discussion applies to $ACE_{ss}  = 0$ and $\varepsilon_1  = 1.$ \citet{Guo:2014} proposed a parametric likelihood ratio test for GPI under Monotonicity and ER. In fact, Proposition \ref{thm::sensitivityanalysis-monotonicity} implies tests for compatibility of GPI and ER, which sometimes can be an important initial step in empirical studies when we are unsure about the underlying assumptions. 
For instance, if we have important covariate information and impose GPI, then Proposition \ref{thm::sensitivityanalysis-monotonicity} allows us to test ER in the noncompliance setting or the causal necessity in the surrogate problem. If the test is rejected, then we may reject ER and causal necessity, but we may also doubt about the GPI assumption. Although this kind of discordant result does not provide a definite answer, it does warn us of the underlying assumptions and may lead us to conduct more careful analysis or better study design. See \citet{Yang:2014} for a concrete example and philosophical discussion about checking compatibility of untestable assumptions in causal inference.

\subsection{Sensitivity Analysis for Monotonicity}\label{subsec:monosen}
Without Monotonicity, we have all four principal strata, and we cannot even identify their proportions without further assumptions.
The inferential difficulties restricts the scope of the current literature to be under Monotonicity.
Some exceptions \citep[e.g.,][]{Zhang:2009, Ding:2011, Frumento:2012} rely on either strong modeling assumptions or additional information.
We take an alternative perspective, and propose an approach to performing sensitivity analysis when Monotonicity is not plausible. We introduce the following sensitivity parameter $\xi$ capturing the deviation from Monotonicity:
$$
\xi = \frac{\Pr(U=\bar{s}s \mid \bm{X})}{\Pr(U = s\bar{s}\mid \bm{X})} ,
$$
which is the ratio between the probabilities of strata $U=\bar{s}s$ and $U=s\bar{s}$ conditional on covariates $\bm{X}.$ Furthermore, the conditional ratio is also the marginal ratio of the probabilities, i.e., $\xi = \Pr(U=\bar{s}s) / \Pr(U=s\bar{s})$. The sensitivity parameter can take values from $0$ to $\infty.$ When $\xi =  0$, we have Monotonicity; when $\xi = 1$, we have equal proportions of $s\bar{s}$ and $\bar{s}s$, and thus zero average causal effect on $S$; when $0 < \xi < 1$, we allow deviation from Monotonicity but still preserve positive average causal effect on $S$; when $\xi > 1$, we have negative average causal effect on $S.$
Without loss of generality, we will assume $p_1-p_0 \geq  0$ and $0\leq  \xi \leq 1$ from now on for sensitivity analysis.

\begin{theorem}
\label{thm::prop-no-monotonicity}
For a fixed sensitivity parameter $\xi$, we can identify the proportions by
\begin{eqnarray}
\pi_{s\bar{s}}  =  \frac{ p_1-p_0} { 1-\xi}, \quad
\pi_{\bar{s}\bar{s}}  = 1-p_0 - \frac{p_1-p_0}{ 1-\xi} ,\quad
\pi_{ss}  = p_1 - \frac{p_1-p_0}{1-\xi}, \quad
\pi_{\bar{s}s}  = \frac{ \xi (p_1-p_0) } { 1-\xi } ,
\label{eq::prop-no-monotonicity}
\end{eqnarray}
which further imply that the sensitivity parameter $\xi$ is bounded by
\begin{equation}\label{eq::sensitivity-range}
0 \leq \xi \leq 1 - \frac{ p_1 - p_0 }{ \min(  p_1, 1-p_0  ) } \leq 1 .
\end{equation}
\end{theorem}

Although $\xi$ is not identifiable, the observed data provide an upper bound for it when the average causal effect on $S$ is non-negative.
Therefore, we need to only perform sensitivity analysis within the empirical version of the above bounds of $\xi$.

Analogously, we can show that the principal score $e_u(\bm{X})$ is identifiable with a known $\xi$, by replacing $p_1$ and $p_0$ in formula \eqref{eq::prop-no-monotonicity} by $p_1( \bm{X})$ and $p_0( \bm{X})$, respectively.
Consequently, GPI is sufficient to identify the PCEs.

\begin{theorem}
\label{thm::PCE-without-monotonicity}
Under GPI, and for a fixed $\xi$, we can identify the PCEs by
\begin{eqnarray*}
ACE_{u}  &=& E\{   w_{1 , u }(\bm{X})  Y\mid Z=1, S=s(1)   \}  - E\{   w_{ 0 , u} (\bm{X} ) Y\mid Z=0,S=s(0)  \} ,
\end{eqnarray*}
where $s(1)$ and $s(0)$ correspond to the values of $S(1)$ and $S(0)$ of $U=u$, and the weight $w_{z,u}(\bm{X})$ is defined in the same way as Proposition \ref{thm::identifiability-monotonicity-GPI}.
\end{theorem}

Proposition \ref{thm::PCE-without-monotonicity} is similar to Proposition \ref{thm::identifiability-monotonicity-GPI}, except for that all the observed groups defined by $(Z,S)$ are mixtures of two latent strata.

To end this subsection, we discuss a model strategy for principal scores without Monotonicity.
Combining $s\bar{s}$ and $\bar{s}s$ into one category, we define $V_i = U_i$ if $U_i=ss$ or $\bar{s}\bar{s}$, and $V_i = s\&\bar{s}$ if $U_i=s\bar{s}$ or $\bar{s}s$.
We can model $\Pr(U\mid \bm{X})$ by two steps. First, we model $\Pr(V\mid \bm{X})$ as a three-level Multinomial Logistic regression. Second, we partition the category of $V$, $s\&\bar{s}$, into two sub-categories of $U$, $s\bar{s}$ and $\bar{s}s$, with probabilities $\Pr(U=s\bar{s} \mid V=s\&\bar{s} , \bm{X}) = 1/(1+\xi)$ and $\Pr(U=\bar{s}s\mid V=s\&\bar{s}, \bm{X}) = \xi/(1+\xi)$.
We show in the supplementary material the EM algorithm for computing the MLE of the above model.
After estimating the principal scores, we can apply the weighting and covariate-adjustment method to estimate the PCEs as discussed in Section \ref{sec:estimation}.

\section{Simulation Studies}\label{sec:simu}

To examine the finite sample performance of our estimators, we conduct a series of simulation studies. Let the sample sizes be $500$ in all scenarios. For unit $i$, we generate $X_{i1}, \ldots, X_{i4} \iidsim N(0,1)$ and $X_{i5} \sim \mathrm{Bern}(1/2)$, and let $\bm X_i = (1, X_{i1}, \ldots, X_{i5})^\top $. We conduct simulations under Strong Monotonicity and Monotonicity, respectively. In each scenario we consider five cases indexed by the parameter $\theta=-1,-0.5,0,0.5,$ and $1$. We postpone the interpretation of $\theta$ until afterwards.

Under Strong Monotonicity, for each $\theta$ we generate principal strata from a Logit model
$
\mathrm{logit} ~ \mathrm{Pr}(U_i = s\bar{s}\mid \bm X_i)=\bm\theta^\top\bm X_i,
$
where $\bm{\theta}=(0,0.5,0.5,1,1,\theta)^\top$. We generate Normal potential outcomes from
$
Y_i(1)\mid \bm X_i \sim  N\left(\sum_{j=1}^5 X_{ij} + 2\cdot I_{\{U=s\bar{s}\}}+1, 1\right)
$
and
$
Y_i(0)\mid \bm X_i \sim  N\left(\sum_{j=1}^5 X_{ij} + 2, 1\right) ;
$
Bernoulli potential outcomes from 
$
\mathrm{logit}~ \mathrm{Pr}\left\{ Y_i(1) = 1\mid \bm X_i\right\}  = 0.3\sum_{j=1}^5 X_{ij}+I_{\{U=s\bar{s}\}}
$
and
$
\mathrm{logit}~ \mathrm{Pr}\left(\{ Y_i(0) = 1\mid \bm X_i\right\} = 0.3\sum_{j=1}^5 X_{ij}+0.5.
$

Under Monotonicity, for each $\theta$ we generate principal strata from a Multinomial Logit model
$
\Pr(U_i=u\mid \bm X_i) = \mathrm{exp}(\bm\theta_u^\top \bm X_i) / \sum_{u^\prime} \mathrm{exp}(\bm\theta_{u^\prime}^\top \bm X_i) 
$
for $u=s\bar{s}, ss, \bar{s}\bar{s}$, where $\bm\theta_{ss}=(0.25,0.5,0.5,1,1, \theta)$, $\bm\theta_{\bar{s}\bar{s}}=(-0.25,1,1,0.5,0.5, \theta)$ and $\bm\theta_{s\bar{s}}=\bm 0$. We generate Normal potential outcomes from
$
Y_i(1)\mid \bm X_i \sim  N\left(\sum_{j=1}^5 X_{ij} - I_{\{U=\bar{s}\bar{s}\}} + 4, 1\right)
$
and
$
Y_i(0)\mid \bm X_i \sim  N\left(\sum_{j=1}^5 X_{ij} + I_{\{U=ss\}} + 1, 1\right) ; 
$
Bernoulli potential outcomes from
$
\mathrm{logit}~\mathrm{Pr}\left\{ Y_i(1) = 1\mid \bm X_i\right\} = 0.3\sum_{j=1}^5 X_{ij}+0.25 \left(I_{\{U=\bar{s}\bar{s}\}}-1\right) 
$
and
$
\mathrm{logit}~\mathrm{Pr}\left\{ Y_i(0) = 1\mid \bm X_i\right\} = 0.3\sum_{j=1}^5 X_{ij}+0.25 \left(1-I_{\{U=ss\}}\right).
$
Although the above data generating mechanisms violate GPI, they satisfy its weaker version, i.e., $\varepsilon_1=\varepsilon_0=1$, which also suffices to ensure Proposition \ref{thm::identifiability-monotonicity-GPI}.

To examine the performance of our estimators with and without (the weaker version of) GPI, in each simulation scenario we analyze the data with and without the binary covariate $X_{i5}$, and we respectively label the corresponding results as ``oracle'' and ``obs''. Without using $X_{i5}$, we can view $\theta$ as a measure of the violation from GPI. 
In Figure \ref{fg:simu}, we present only the results for $ACE_{\bar{s}\bar{s}}$ using the model-assisted estimator, because in our simulations the naive weighting estimators are uniformly worse in terms of estimation efficiency. For the ease of presentation, we omit similar results for other principal strata. We use $500$ bootstraps to construct $95\%$ confidence intervals, and focus on the average biases and coverage rates over $1000$ repeated samplings. With the binary covariate, our estimator has small biases and achieves nominal coverage rates, for both Normal and Bernoulli potential outcomes. Without the binary covariate, our estimators have bias issues for both Normal and Bernoulli potential outcomes when $|\theta|$ approaches one, i.e., GPI is severely violated. The interval estimates under coverage the true parameters for Normal outcomes when $|\theta|$ approaches one, but the coverage properties for Bernoulli outcomes are robust with respect to the violations of GPI. 
This bias issue, as well as the untestable nature of PI and GPI, warns us that sensitivity analysis with respect to PI and GPI, as proposed in Section \ref{subsec:GPIsen}, must be an essential part of any empirical studies using principal scores to analyze principal stratification problems.

Due to the constraint of space, we compare our model-assisted estimator with \citet{Jo:2009}'s model-based estimator in the supplementary material, showing that our estimator does not lose efficiency compared to full modeling and is robust to model misspecification of the outcome.

\begin{figure}[htbp]
\centering
\begin{subfigure}{\textwidth}
  \centering
  \includegraphics[height=.5\linewidth, width=.7\linewidth]{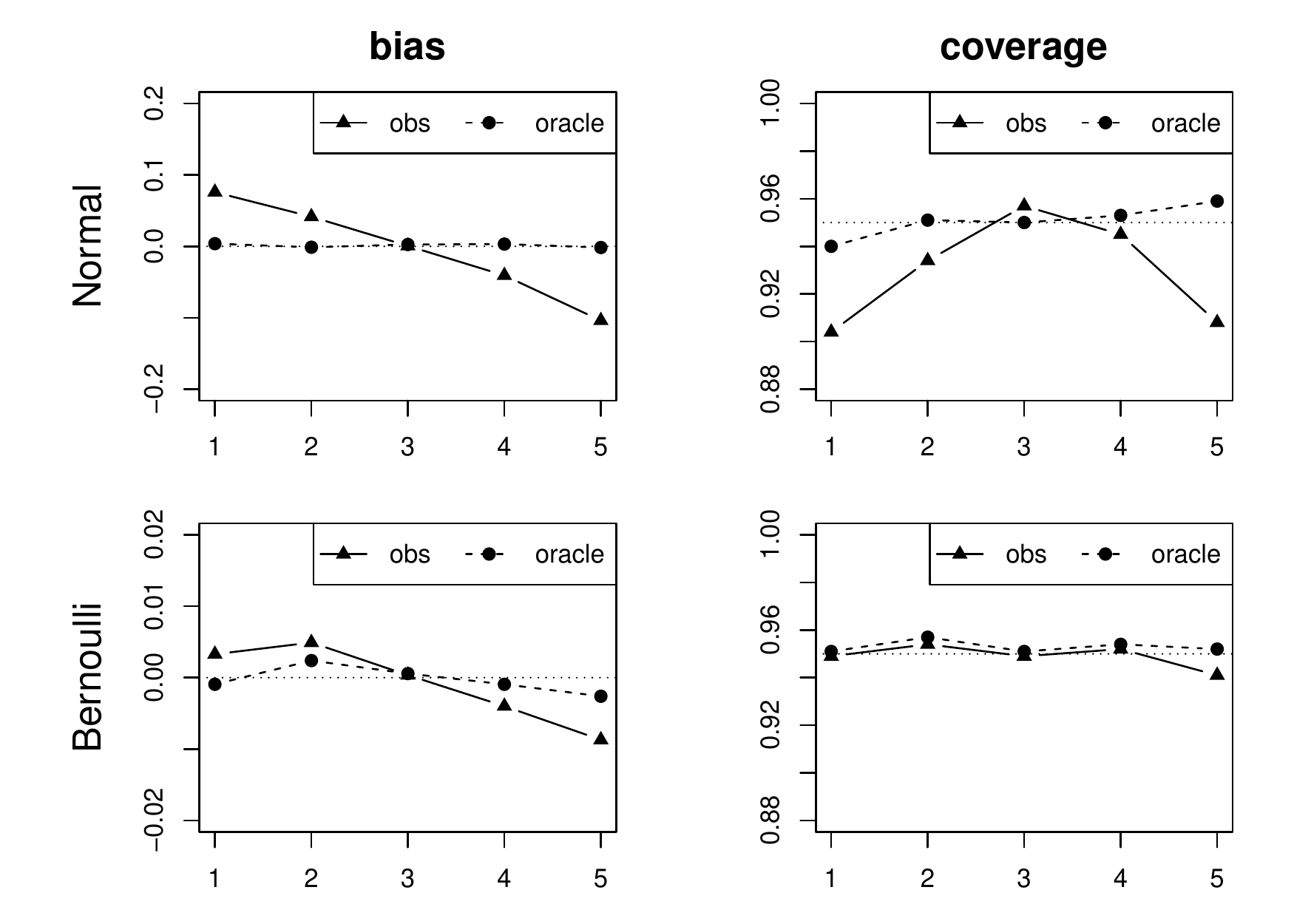}
  \caption{Under Strong Monotonicity. The horizontal axis shows the case numbers, and ``obs'' and ``oracle'' denote the cases with and without the binary covariate, respectively.}
  \label{fg:sm2}
\end{subfigure}
\begin{subfigure}{\textwidth}
  \centering
  \includegraphics[height=.5\linewidth, width=.7\linewidth]{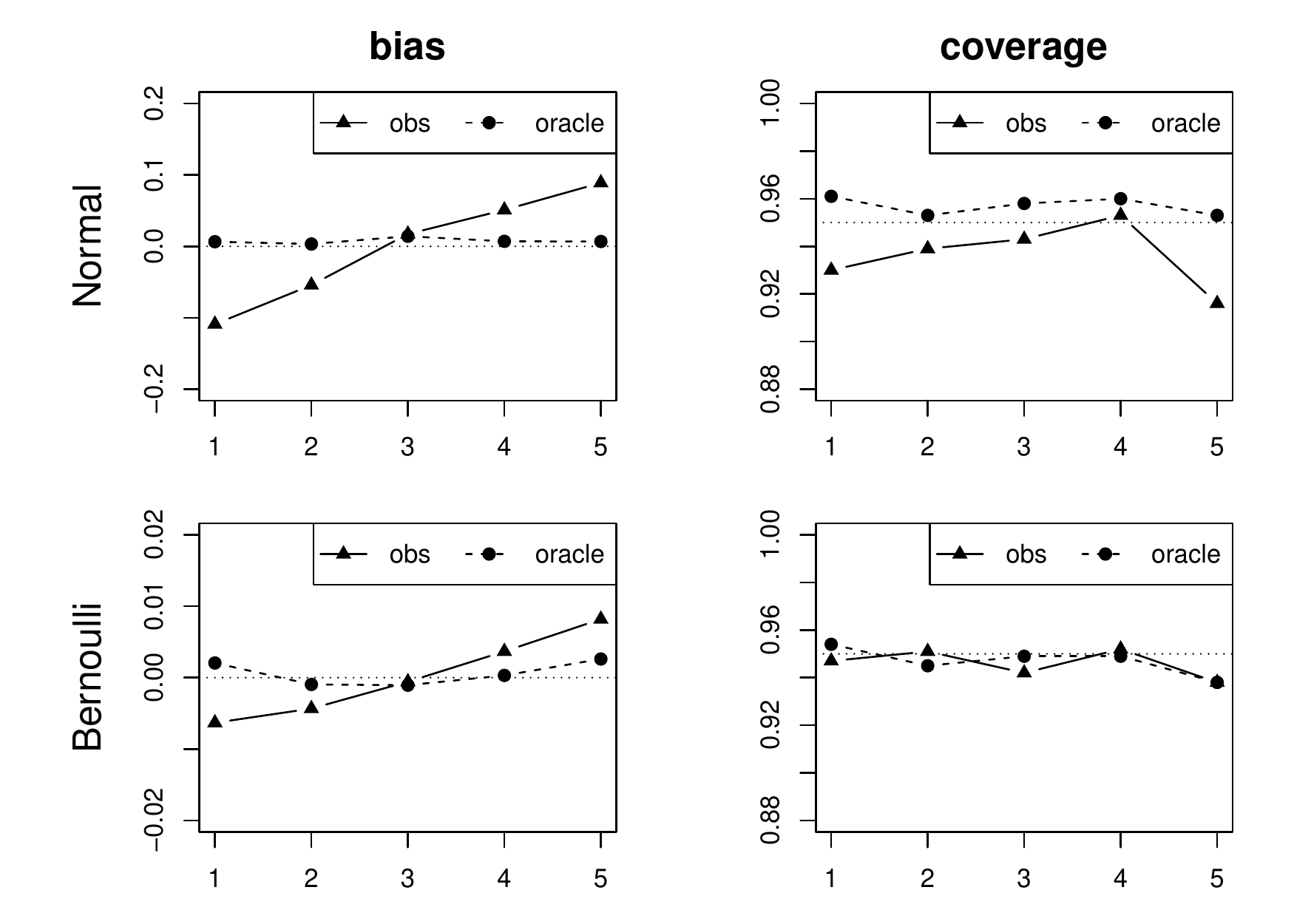}
  \caption{Under Monotonicity. The horizontal axis shows the case numbers, and ``obs'' and ``oracle'' denote the cases with and without the binary covariate, respectively.}
  \label{fg:m2}
\end{subfigure}
\caption{Simulation Results for $ ACE_{\bar{s}\bar{s}}$. Each subfigure is a $2\times2$ matrix summarizing two repeated sampling properties (average biases and coverage rates of interval estimates).}
\label{fg:simu}
\end{figure}

\section{Applications}\label{sec:apply}

\subsection{An Encouragement Experiment with Noncompliance}\label{subsec::flu}
In this section, we re-analyze a flu shot encouragement experiment data set previously studied by \cite{Hirano:2000}. Between 1978 to 1980, a general medicine clinic in Indiana conducted an encouragement experiment, in which participating individuals' physicians were randomly assigned to the treatment arm with computer-generated letters encouraging them to inoculate their patients, or the control arm with no letters. The outcome of interest is the individual's flu-related hospitalization status during the subsequent winter.
As in \cite{Hirano:2000}, we use the data from 1980, with $2893$ experimental units. In our analysis, $Z=1$ if an individual's physician received the letter, and $Z=0$ otherwise. The intermediate variable $S=1$ if the individual received the flu shot, and $S=0$ otherwise. The outcome of interest $Y=1$ if the individual was hospitalized for flu-related reasons, and $Y=0$ otherwise. The Monotonicity assumption is plausible for this data set, because we expect the encouragement letter to have nonnegative effect on taking the flu shot. Because this encouragement experiment is an open-label trial, previous researchers doubted ER due to the possible ``direct effect'' of the flu shot encouragement on the outcome.

To start our analysis, we use the covariate balancing conditions in Corollary \ref{corollary::M} to check the plausibility of the Logistic principal score model. Choosing $h(\bm X)=\bm X$ is reasonable, because all covariates are binary except for ``age.'' The balance checking is equivalent to estimating the PCEs on $h(\bm{X})$, known to be zero. Therefore, the corresponding ``standardized $t$-statistics'' should follow standard Normal distributions. Figure \ref{fg:flu1} shows that the covariates are well balanced. Assuming GPI, we estimate the PCEs with standard errors and $95\%$ confidence intervals in parentheses as:
$$
\widehat{ACE_{s\bar{s}}}=-0.018\  [-0.052,0.016] , \quad 
\widehat{ACE_{ss}}=-0.046\  [-0.091,0.002] , \quad 
\widehat{ACE_{\bar{s}\bar{s}}}=-0.006\     [-0.030,0.017].
$$

Therefore, for compliers, receiving the encouragement letter will lower the chance of flu related hospital visit by $1.8\%$, but this effect is not significant. Furthermore, ER seems plausible for never-takers. However, there is some evidence that it does not hold for always-takers, because the upper confidence limit is close to zero. Our findings corroborate \cite{Hirano:2000}'s argument that ``it is more plausible to impose the exclusion restriction for never-takers than for always-takers.'' \citet{Hirano:2000}'s results required careful analysis, including using data dependent priors with several tuning parameters that account for the background knowledge. Our analysis under GPI yields coherent conclusions as theirs. Therefore, if we believe their prior knowledge and statistical analysis, then GPI seems plausible in this example. At least, there is no obvious contradiction derived from two different analysis, and our results under GPI have meaningful scientific interpretations.

Nevertheless, the data cannot validate GPI, an untestable assumption requiring observed covariates $\bm{X}$ contain all characteristics related to the latent principal stratum and potential outcomes; \citet{Hirano:2000}'s analysis does not contradict GPI but does not prove it either. As advocated in Section \ref{sec:sen}, we perform sensitivity analysis for GPI,  allowing $(\varepsilon_1,\varepsilon_0) $ to vary within $[1/2, 2]\times [1/2,2]$ with results in Figure \ref{fg:flu2}. If we are willing to assume that never-takers are the strongest patients and the always-takers are the weakest patients, we can restrict our sensitivity analysis within the region with $\varepsilon_1<1$ and $\varepsilon_0>1.$ Interestingly, within this range most of the confidence intervals $\widehat{ACE_{s\bar{s}}}$ do not cover zero, suggesting that there is a significant causal effect for compliers. Furthermore, $\widehat{ACE_{ss}}$ and $\widehat{ACE_{\bar{s}{s}}}$ are relatively robust to $\varepsilon_1$ and $\varepsilon_0,$ respectively. The upper confidence limits for $\widehat{ACE_{ss}}$ are always close to zero as $\varepsilon_1$ varies, showing weak evidence for violation of ER for always-takers; the centers of the confidence intervals for $\widehat{ACE_{\bar{s}\bar{s}}}$ are always close to zero as $\varepsilon_0$ varies, suggesting that ER holds for never-takers. Fortunately, although the point and interval estimators vary with the sensitivity parameters, the final conclusions do not change materially.

\begin{figure}[htbp]
\centering
\begin{subfigure}{\textwidth}
  \centering
  \includegraphics[height=.4\linewidth, width=.58\linewidth]{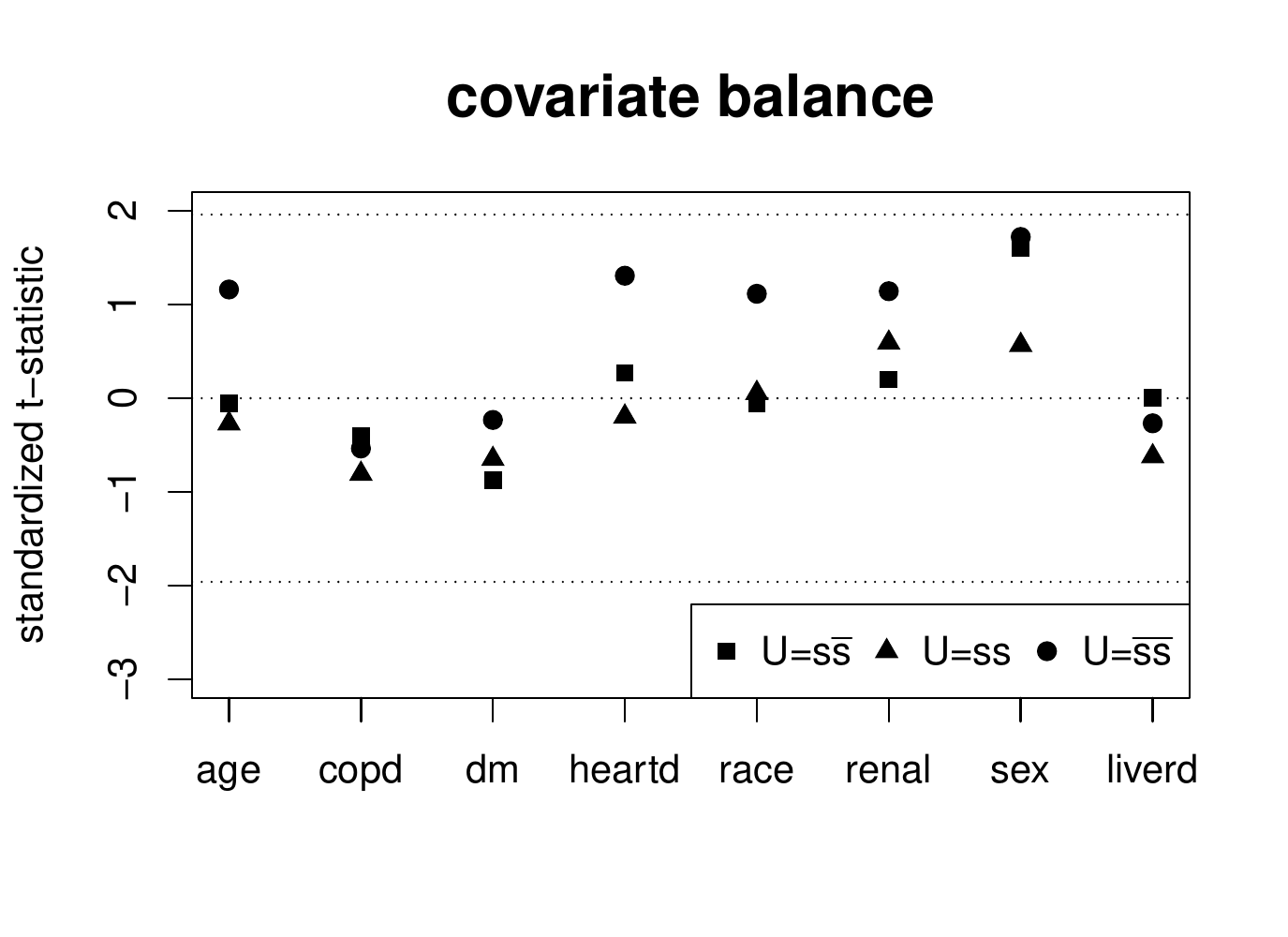}
  \caption{Covariate Balance Check. The horizontal axis shows the names of the covariates.}
  \label{fg:flu1}
\end{subfigure}
\begin{subfigure}{\textwidth}
  \centering
  \includegraphics[height=.4\linewidth, width=\linewidth]{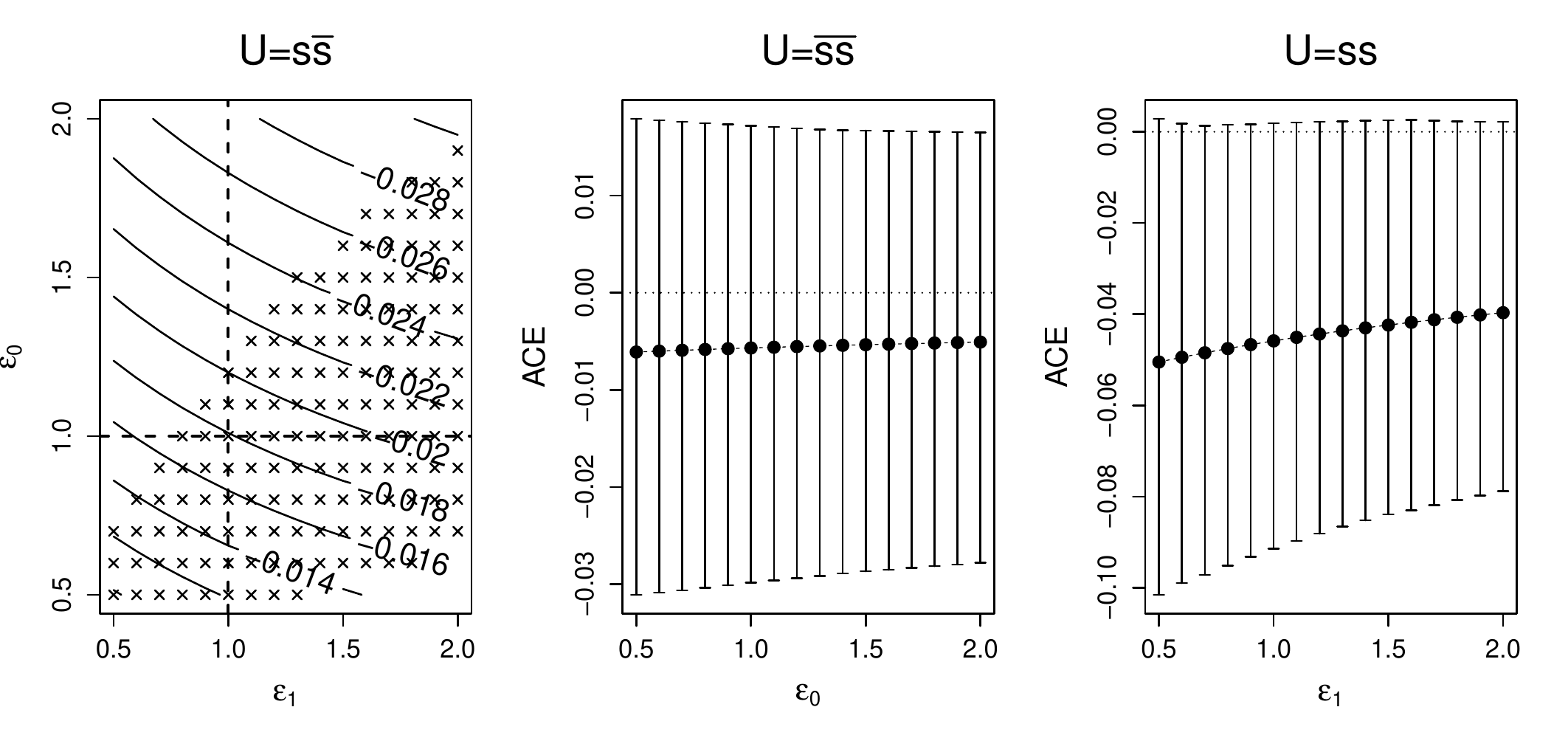}
  \caption{Sensitivity Analysis for GPI. The first subfigure shows the contours of the point estimates of $ACE_{s\bar{s}}$ for fixed values of $\varepsilon_1$ and $\varepsilon_0$, where ``$\times$'' denotes $(\varepsilon_1, \varepsilon_0)$ such that the corresponding interval estimate covers 0. The second and third subfigures show the point and interval estimates of $ACE_{\bar{s}\bar{s}}$ for fixed values of $\varepsilon_0$, and the point and interval estimates of $ACE_{ss}$ for fixed values of $\varepsilon_1$, respectively. }
  \label{fg:flu2}
\end{subfigure}
\caption{The Flu Shot Encouragement Experiment}
\label{fg:flu}
\end{figure}

\subsection{A Randomized Trial with Truncation by Death} \label{subsec::swog}
From  October 1999 to January 2003, the Southwest Oncology Group (SWOG) conducted a randomized phase III trial (protocol 99-16) to compare the treatment of docetaxel and estramustine (DE) with mitoxantrone and prednisone (MP) in patients with metastatic, androgen-independent prostate cancer (\citealp{Petrylak:2004}).
A total of $674$ eligible patients participated in the study between October 1999 and January 2003.  Study participants  were randomly assigned to the DE arm  or the MP arm. The primary outcome is the survival time, and the secondary outcome is the health related quality of life (HRQOL). \cite{Petrylak:2004} have reported the overall survival benefit of taking DE over taking MP. In our analysis, we are interested in assessing the causal effect of DE versus MP on the HRQOL one year after receiving the treatment. In our analysis, $Z=1$ if a patient received DE, and $Z=0$ if the patient received MP. We use the difference between HRQOL after one year and the baseline HRQOL as the outcome of interest. The survival indicator $S=1$ if a patient survived after one year.

Because of the truncation by death problem, we are interested in estimating the survivor average causal effect.
As in the previous example, we first check the plausibility of the Logistic principal score model. Figure \ref{fg:swog1} shows that we achieve covariate balance. The point estimate of the SACE is $3.07$, but its standard error is $2.976$ and the $95\%$ confidence interval $[-2.93, 8.69]$ covers zero. The results show that DE is not significantly more effective than MP to improve the HRQOL of the patients, which is similar to the analysis in \cite{Ding:2011}. However, applying \citet{Zhang:2009}'s Normal mixture model, we obtain point estimate $12.34$ with standard error $47.17$. The tremendous variability of the estimator is due to the unstable numerical issue and unreliable large sample Normal approximation, as investigated by \citet{Mealli:2015}.

However, both the treatment and control are active drugs for the prostate cancer, and therefore it is not reasonable to assume that the treatment is more effective than the control for all patients, i.e., Monotonicity may not hold.
We perform sensitivity analysis for Monotonicity, and choose the range of the sensitivity parameter $\xi$ based on Proposition \ref{thm::prop-no-monotonicity}. We compute from the data that $\widehat{p}_1=0.496$ and $\widehat{p}_0=0.389$, and therefore $0\leq \widehat{\xi} \leq 1 - ( \widehat{p}_1 - \widehat{p}_0  ) / \{  \min(  \widehat{p}_1, 1-\widehat{p}_0  ) \} \approx 0.217.$
The sensitivity analysis results in Figure \ref{fg:swog2} show that the point and interval estimates of $ACE_{ss}$ are relatively robust to $\xi.$ Furthermore, the interval estimates for $\widehat{ACE_{ss}}$ always cover zero as $\xi$ varies. In summary, the sensitivity analysis results confirm our previous conclusions.

\begin{figure}[t!]
\centering
\begin{subfigure}[t]{.45\textwidth}
  \centering
  \includegraphics[height=\linewidth, width=\linewidth]{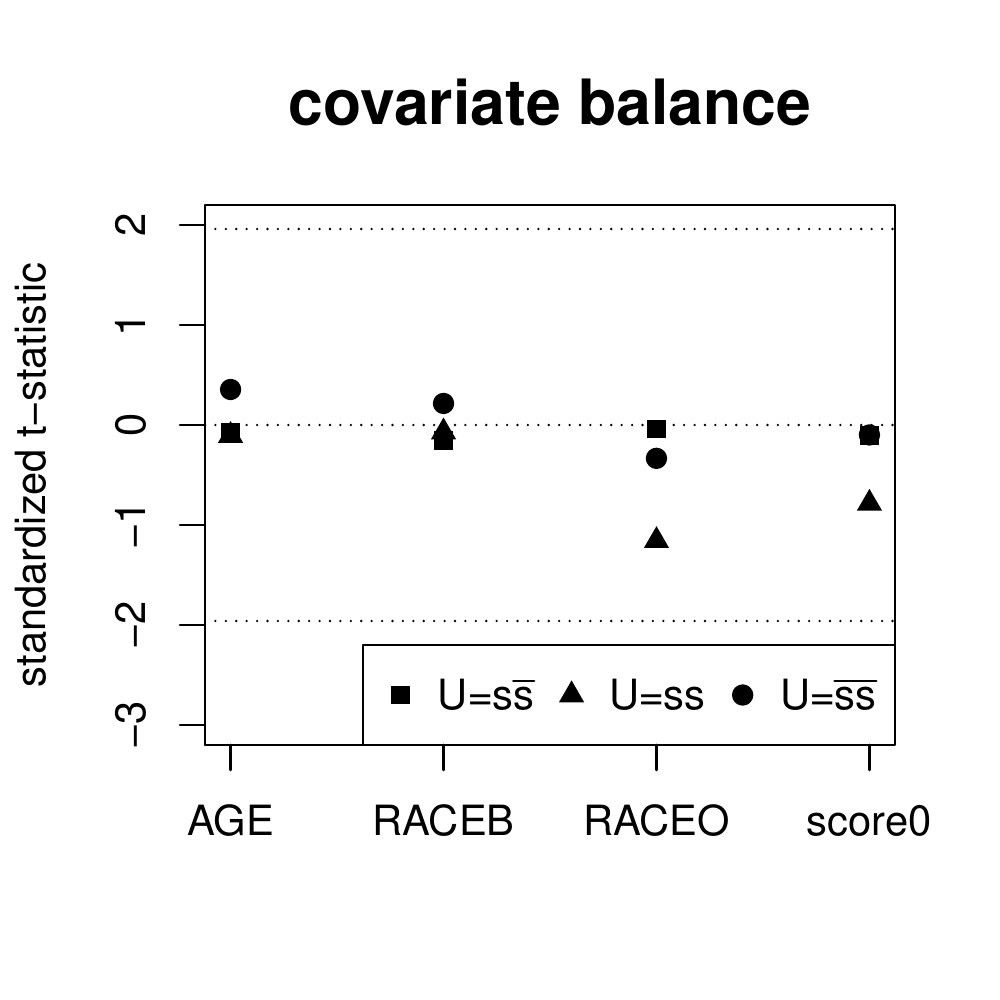}
  \caption{Covariate Balance Check. The horizontal axis shows the names of the covariates.}
  \label{fg:swog1}
\end{subfigure}
\begin{subfigure}[t]{.45\textwidth}
  \centering
  \includegraphics[height=\linewidth, width=\linewidth]{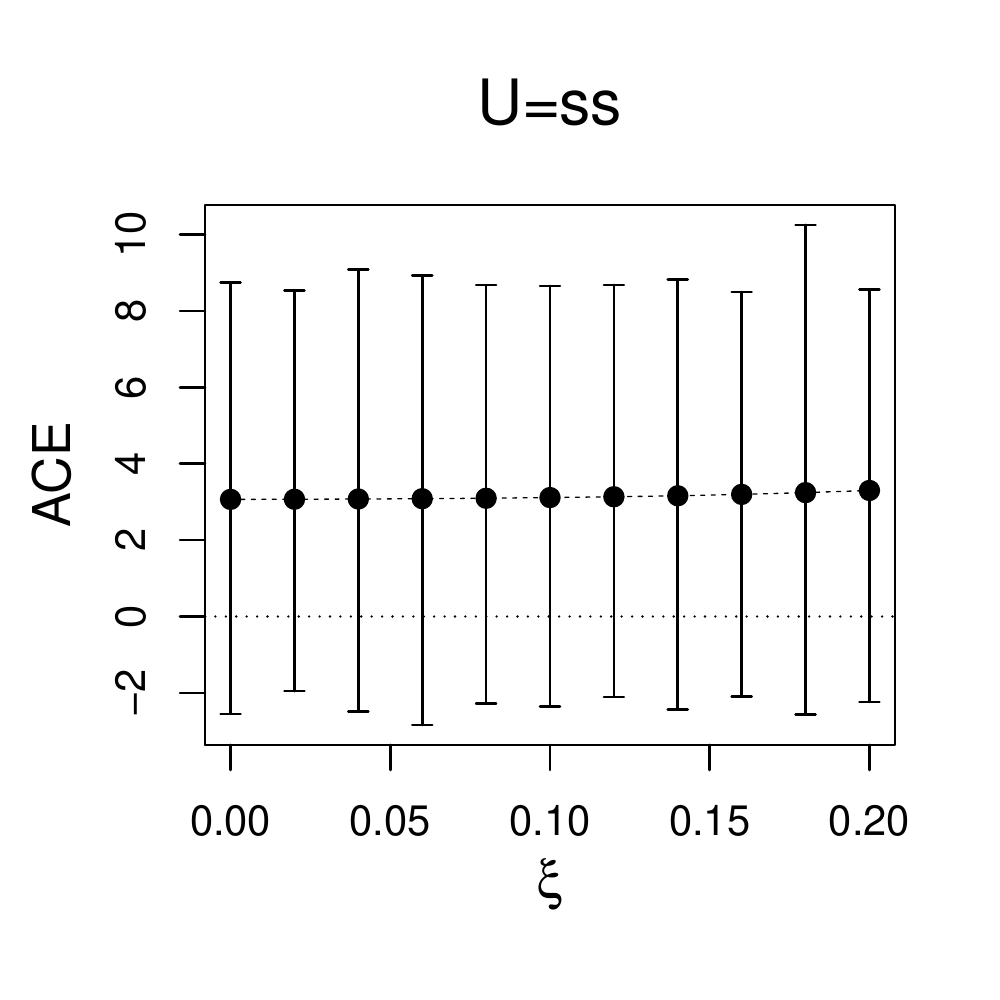}
  \caption{Sensitivity Analysis for Monotonicity. We show the point and interval estimates of $ACE_{ss}$ for fixed values of $\xi$. }
  \label{fg:swog2}
\end{subfigure}
\caption{The SWOG Randomized Trial}
\label{fg:swog}
\end{figure}

\section{Discussion}
\label{sec::discussion}

In observational studies, causal effects can be estimated by inverse propensity score weighting \citep{Rosenbaum:1983biometrika}, which may be numerically unstable and have poor finite sample properties. Our estimators, weighted by probabilities themselves, do not suffer from these problems. Researchers \citep[e.g.,][]{Bang:2005} have developed doubly robust methods in observational studies. Similar to our model-assisted estimators, these doubly robust estimators can also be derived from regression estimators in surveys \citep{Cochran:1977}. Because of this similarity, it will be interesting to develop doubly robust estimators under the PI assumptions that are consistent when either the principal score or the outcome model is correctly specified.

The theoretical results have demonstrated the two-fold role of the pretreatment covariates. First, the plausibility of the ignorability assumptions rely crucially on adequate covariates. Second, with more covariates that are predictive to the outcome, the covariate-adjusted estimators will be more efficient. Our results suggest that, in the design of randomized experiments, it is important for practitioners to try their best to collect covariates that are predictive to both the latent principal strata and the potential outcomes, which echoes \citet{Jo:2009} and \citet{Mealli:2013}.

Although in the main text we focused on the average causal effect within principal strata, our results can be easily extended to general causal measures. For example, we can dichotomize the outcome to identify the distributional causal effects \citep{Ju:2010}. For binary $S$, we have derived clean results and easy-to-implement estimators. For general discrete or continuous $S$, we can likewise derive theoretical results under PI by modifying the weights in Propositions \ref{thm::identifiability-monotonicity-GPI} and \ref{thm::prop-no-monotonicity}. However, a continuous $S$ results in infinitely many principal strata, which makes it challenging to estimate the principal scores and outcome distributions conditional on continuous variables. We need more structural assumptions on the causal problems \citep{Jin:2008, Schwartz:2011} and more sophisticated statistical inferential tools.

Missing data is an important problem that often arises in real data analysis. Our two-step procedure has some advantages if only some outcomes are missing. We can conduct the first step for estimating principal scores without any difficulty, and need only to modify the second weighting step. If the outcome is missing at random, then we can simply weight each observation by the inverse of the conditional probability of being observed given $(Z,S,\bm{X})$. 
However, for missing data problem, the key issue is the missing data mechanism. Other missing data mechanisms, e.g., latent ignorability \citep{Frangakis:1999}, may be more plausible, but the identification becomes challenging. Due to this complication, we leave the missing data problem for future research.

\section*{Acknowledgments}

Peng Ding's work is partially supported by grant R305D150040 from Institute of Education Sciences, USA.
We are grateful for the comments from Professors Donald Rubin and Tirthankar Dasgupta, and other participants in the ``Matched Sampling and Study Designs'' seminar at Harvard. We benefit from the suggestions of Avi Feller at Berkeley, Keli Liu at Stanford, and Professors Luke W. Miratrix and Joseph K. Blitzstein at Harvard. The comments from the Joint Editor, the Associate Editor, and two reviewers have helped improve the quality of our paper significantly.

\bibliographystyle{apalike}
\bibliography{PSdesign2015}

\newpage
\begin{center}
{\bf \Huge Supplementary material}
\end{center}
\bigskip

\renewcommand {\thesection} {A.\arabic{section}}
\renewcommand {\theequation} {A.\arabic{equation}}
\renewcommand {\thelemma} {A.\arabic{lemma}}

\ref{sec::sm} contains the proofs under Strong Monotonicity. \ref{sec::mono} contains the proofs under Monotonicity. \ref{sec::no-mono} contains the proofs without Monotonicity. \ref{sec::principal-scores-computations} presents the computational details of the EM algorithms for the principal score models. \ref{sec::weighting-estimators} gives the explicit forms of the weighting and model-assisted estimators. \ref{sec::additional simulation} contains additional simulations comparing our model-assisted estimator with the model-based estimator in Jo and Stuart (2009).
We will use $f(\cdot)$ as for a (conditional) probability density function, and the following basic identity of importance sampling to simplify our proofs.

\begin{lemma}
\label{lemma::importance-sampling}
Assuming existence of moments, $X\sim f_1(x)$ and $Y\sim f_2(y)$, we have
\begin{eqnarray}
\label{eq::importance-sampling}
E \{ g(X) \} = E \left\{  \frac{f_1(Y)}{f_2(Y)}   g(Y)   \right\}.
\end{eqnarray}
\end{lemma}

\section{Proofs of the Propositions Under Strong Monotonicity}
\label{sec::sm}

To prove Proposition 1, we first need the following
lemmas.

\begin{lemma}
\label{lemma::SM-balance}
Under Strong Monotonicity, $\bm{X}\ind U\mid e_u(\bm{X})$ for $u=s\bar{s}$ and $\bar{s}\bar{s}$.
\end{lemma}

\begin{proof}[Proof of Lemma \ref{lemma::SM-balance}]
We have 
$
\Pr\{ U=s\bar{s}\mid \bm{X}, e_{s\bar{s}}(\bm{X}) \} 
=\Pr(U=s\bar{s}\mid \bm{X}) = e_{s\bar{s}}(\bm{X}), 
$
and 
\begin{eqnarray*}
\Pr\{ U=s\bar{s}\mid e_{s\bar{s}}(\bm{X})  \} 
= E[   \Pr\{  U=s\bar{s} \mid \bm{X}, e_{s\bar{s}}(\bm{X}) \}  \mid e_{s\bar{s}}(\bm{X}) ] = E\{ e_{s\bar{s}}(\bm{X}) \mid e_{s\bar{s}}(\bm{X})\} = e_{s\bar{s}}(\bm{X}).
\end{eqnarray*}
Therefore, $ \Pr\{ U=s\bar{s}\mid \bm{X}, e_{s\bar{s}}(\bm{X}) \} = \Pr\{ U=s\bar{s}\mid e_{s\bar{s}}(\bm{X})  \} $, implying $\bm{X}\ind U\mid e_{s\bar{s}}(\bm{X})$. Because $e_{\bar{s}\bar{s}}(\bm{X})  =1- e_{s\bar{s}}(\bm{X})$, other conditional independence also follows.
\end{proof}

\begin{lemma}
\label{lemma::SM-sufficiency}
Under Strong Monotonicity and PI, $Y(0)\ind U \mid e_u(\bm{X})$ for $u=s\bar{s}$ and $\bar{s}\bar{s}$.
\end{lemma}

\begin{proof}[Proof of Lemma \ref{lemma::SM-sufficiency}]
Applying Law of Iterated Expectation (LIE), we have
\begin{eqnarray*}
\Pr\{ U = s\bar{s} \mid Y(0), e_{s\bar{s}}(\bm{X}) \}
&=& E [   \Pr\{ U = s\bar{s} \mid Y(0), e_{s\bar{s}}(\bm{X}) , \bm{X} \}  \mid   Y(0), e_{s\bar{s}}(\bm{X})     ]\\
&=& E [   \Pr\{ U = s\bar{s} \mid Y(0),  \bm{X} \}  \mid   Y(0), e_{s\bar{s}}(\bm{X})     ],
\end{eqnarray*}
which, by PI, reduces to
$
E [   \Pr\{ U = s\bar{s} \mid  \bm{X} \}  \mid   Y(0), e_{s\bar{s}}(\bm{X})     ]
= E\{ e_{s\bar{s}}(\bm{X})  \mid   Y(0), e_{s\bar{s}}(\bm{X})  \}   = e_{s\bar{s}}(\bm{X}).
$
According to the proof of Lemma \ref{lemma::SM-balance}, we also have $\Pr\{ U=s\bar{s} \mid e_{s\bar{s}}(\bm{X})\}  = e_{s\bar{s}}(\bm{X})$, and therefore $\Pr\{ U = s\bar{s} \mid Y(0), e_{s\bar{s}}(\bm{X}) \}  = \Pr\{ U=s\bar{s} \mid e_{s\bar{s}}(\bm{X})\}$, implying $Y(0)\ind U\mid e_{s\bar{s}}(\bm{X})$. Other conclusions about conditional independence also hold.
\end{proof}

\begin{proof}[Proof of Proposition 1]
In the treatment group, $(Z_i=1,S_i=1)$ is equivalent to $(Z_i=1,U_i=s\bar{s})$, and $(Z_i=1,S_i=0)$ is equivalent to $(Z_i=1,U_i=\bar{s}\bar{s})$. Therefore, it is straightforward to identify
$$
E\{ Y(1)\mid U=s\bar{s} \}  = E(Y\mid Z=1,S=1),\quad E\{ Y(1)\mid U=\bar{s}\bar{s} \}  = E(Y\mid Z=1, S=0)
$$
by the observed data.
The control group is a mixture of $U=s\bar{s}$ and $U=\bar{s}\bar{s}$. 

On the one hand, it is relatively easy to show that
\begin{eqnarray}\label{eq::appendix-mean0}
E\{ Y(0)\mid U=u\}  = E[  E\{ Y(0)\mid U=u, e_u(\bm{X})\}   \mid U=u] = E [  E\{ Y(0)\mid e_u( \bm{X}  )\}  \mid U=u],
\end{eqnarray}
according to Lemma \ref{lemma::SM-sufficiency}. On the other hand, the weighted mean is
$$
E\{  w_{u}(\bm{X} ) Y\mid Z=0\} = E\left\{  \frac{e_u(\bm{X} )}{\pi_u} Y(0) \right\}
= E\left[  E\left\{  \frac{e_u(\bm{X} )}{\pi_u} Y(0) \mid e_u(\bm{X}) \right\} \right]
= E  \left[  \frac{e_u(\bm{X})}{\pi_u} E\{ Y(0)\mid e_u(\bm{X}) \}  \right]  .
$$
Because
$$
f\{  e_u(\bm{X}) \mid U=u\} = \frac{ f\{ e_u(\bm{X}) \}  \Pr\{ U=u\mid e_u(\bm{X})  \}   }{ \pi_u } = \frac{ f\{ e_u(\bm{X}) \}  e_u(\bm{X})  }{\pi_u},
$$
according to the proof of Lemma \ref{lemma::SM-balance},
the weighted mean becomes
\begin{eqnarray}\label{eq::appendix-mean0-weighted}
E\{  w_{u}(\bm{X} ) Y\mid Z=0\} = E  \left[    \frac{   f\{  e_u(\bm{X}) \mid U=u\}    }{    f\{  e_u(\bm{X}) \}  }       E\{ Y(0)\mid e_u(\bm{X})  \}  \right].
\end{eqnarray}
Therefore, formulas (\ref{eq::appendix-mean0}) and (\ref{eq::appendix-mean0-weighted}) are tied together by formula (\ref{eq::importance-sampling}) of importance sampling, yielding $E\{ Y(0)\mid U=u\} = E\{  w_{u}(\bm{X} ) Y\mid Z=0\}.$
This completes the proof.
\end{proof}

\begin{proof}
[Proof of Corollary 1]
We can treat $h(\bm{X})$ as a new ``outcome,'' on which the treatment has zero PCEs. The conclusions follow directly from Proposition 1.
\end{proof}

In order to prove Proposition 3, we need an additional lemma.

\begin{lemma}
\label{lemma::SM-sensitivity-lemma}
We can also represent the sensitivity parameter $\varepsilon$ as
\begin{eqnarray}
\label{eq::sensi-sm-equiv}
\varepsilon =   {  E\{ Y(0)\mid U=s\bar{s} , e_u(\bm{X}) \}   \over  E\{ Y(0)\mid U=\bar{s}\bar{s} , e_u(\bm{X}) \}  }  \quad (u=s\bar{s}, \bar{s}\bar{s}).
\end{eqnarray}
\end{lemma}

\begin{proof}[Proof of Lemma \ref{lemma::SM-sensitivity-lemma}]
We have
\begin{eqnarray*}
E\{ Y(0)\mid U=s\bar{s} , e_u(\bm{X}) \}  &=& E[   E\{ Y(0)\mid U=s\bar{s} , e_u(\bm{X}) , \bm{X} \}  \mid U=s\bar{s} , e_u(\bm{X}) ] \\
&=& E[   E\{ Y(0)\mid U=s\bar{s} ,  \bm{X} \}  \mid U=s\bar{s} , e_u(\bm{X}) ]   \\
&=&\varepsilon  E[   E\{ Y(0)\mid U=\bar{s}\bar{s} ,  \bm{X} \}  \mid U=s\bar{s} , e_u(\bm{X}) ] \\
&=& \varepsilon  E[   E\{ Y(0)\mid U=\bar{s}\bar{s} ,  \bm{X} \}  \mid U=\bar{s}\bar{s} , e_u(\bm{X}) ] ,
\end{eqnarray*}
where the last two lines follow from the definition of $\varepsilon$ and Lemma \ref{lemma::SM-balance}. We also have
\begin{eqnarray*}
E\{ Y(0)\mid U=\bar{s}\bar{s} , e_u(\bm{X}) \}  &=& E[   E\{ Y(0)\mid U=\bar{s}\bar{s} , e_u(\bm{X}) , \bm{X}  \}  \mid U=\bar{s}\bar{s} , e_u(\bm{X}) ]\\
&=&E[   E\{ Y(0)\mid U=\bar{s}\bar{s} ,  \bm{X} \}  \mid U=\bar{s}\bar{s} , e_u(\bm{X}) ].
\end{eqnarray*}
Therefore, we can represent the sensitivity parameter $\varepsilon$ as in formula (\ref{eq::sensi-sm-equiv}).
\end{proof}

\begin{proof}[Proof of Proposition 3]
We prove only the conclusion about $ACE_{s\bar{s}}$; conclusions for the others are analogous. We first observe that
\begin{eqnarray*}
&&E\{ Y(0)\mid e_{s\bar{s}}(\bm{X}) \} \\
&=& E\{ Y(0)\mid e_{s\bar{s}}(\bm{X}) , U = s\bar{s}\}  \Pr\{ U=s\bar{s} \mid  e_{s\bar{s}}(\bm{X})\} + E\{ Y(0)\mid e_{s\bar{s}}(\bm{X}) , U = \bar{s}\bar{s}\}  \Pr\{ U=\bar{s}\bar{s} \mid  e_{s\bar{s}}(\bm{X})\} \\
&=& E\{ Y(0)\mid e_{s\bar{s}}(\bm{X}) , U = s\bar{s}\}  e_{s\bar{s}}(\bm{X})+ E\{ Y(0)\mid e_{s\bar{s}}(\bm{X}) , U = \bar{s}\bar{s}\}  e_{\bar{s}\bar{s}}(\bm{X}) ,
\end{eqnarray*}
according to the proof of Lemma \ref{lemma::SM-balance}. Further, Lemma \ref{lemma::SM-sensitivity-lemma} reduces the above result to
$$
E\{ Y(0)\mid e_{s\bar{s}}(\bm{X}) \}  =  \left\{  e_{s\bar{s}}(\bm{X}) + \frac{ e_{\bar{s}\bar{s}}(\bm{X})} { \varepsilon} \right\}
E\{ Y(0)\mid e_{s\bar{s}}(\bm{X}), U=s\bar{s} \}  ,
$$
which further gives
$$
E\{ Y(0)\mid e_{s\bar{s}}(\bm{X}), U=s\bar{s} \}
=
 \left\{  e_{s\bar{s}}(\bm{X}) + \frac{ e_{\bar{s}\bar{s}}(\bm{X}) } { \varepsilon} \right\} ^{-1}  E\{ Y(0)\mid e_{s\bar{s}}(\bm{X}) \}
 =
  \frac{\varepsilon }{\varepsilon e_{s\bar{s}}(\bm{X}) +e_{\bar{s}\bar{s}}(\bm{X}) }  E\{ Y(0)\mid e_{s\bar{s}}(\bm{X}) \} .
$$
Therefore, on the one hand we have
\begin{eqnarray*}
E\{ Y(0)\mid U=s\bar{s} \}
=  E[ E\{ Y(0)\mid U=s\bar{s}, e_{s\bar{s}}(\bm{X}) \}   \mid U=s\bar{s} ] 
= E\left[   \frac{\varepsilon }{\varepsilon e_{s\bar{s}}(\bm{X}) +e_{\bar{s}\bar{s}}(\bm{X}) }  E\{ Y(0)\mid e_{s\bar{s}}(\bm{X}) \}
\mid U=s\bar{s} \right].
\end{eqnarray*}
On the other hand, the weighted mean can be represented as
$$
E\{  w_{s\bar{s}}^\varepsilon (\bm{X}) Y\mid Z=0 \} / \pi_{s\bar{s}}
= E\left[   \frac{  e_{s\bar{s}}(\bm{X}) }{\pi_{s\bar{s}}}     \frac{\varepsilon}{\varepsilon e_{s\bar{s}}(\bm{X}) +e_{\bar{s}\bar{s}}(\bm{X}) }
E\{ Y(0) \mid e_{s\bar{s}}(\bm{X}) \} \right].
$$
According to the proof of Proposition 1, $e_{s\bar{s}}(\bm{X}) / \pi_{s\bar{s}}=f\{  e_u(\bm{X}) \mid U=u\}   /   f\{  e_u(\bm{X}) \}$ is the importance weight, implying that $E\{ Y(0)\mid U=s\bar{s} \} = E\{  w_{s\bar{s}}^\varepsilon (\bm{X}) Y\mid Z=0 \} / \pi_{s\bar{s}} .$
\end{proof}

\section{Proofs of the Propositions Under Monotonicity}
\label{sec::mono}

To prove Proposition 2, we introduce some lemmas, which rely on the following definitions. 
\begin{eqnarray*}
e_{1, s\bar{s}}(\bm{X}) = \frac{ e_{s\bar{s}}(\bm{X}) }{    e_{s\bar{s}}(\bm{X}) + e_{ss}(\bm{X})   } , 	&\quad&
 e_{1,ss}(\bm{X}) = \frac{  e_{ss}(\bm{X}) }{    e_{s\bar{s}}(\bm{X}) + e_{ss}(\bm{X})  } , \\
 e_{0, s\bar{s}} (\bm{X}) = \frac{  e_{s\bar{s}} (\bm{X}) }{    e_{s\bar{s}}(\bm{X}) + e_{\bar{s}\bar{s}}(\bm{X}) } , &\quad&
e_{0, \bar{s}\bar{s}} (\bm{X}) = \frac{  e_{ \bar{s}\bar{s}} (\bm{X}) }{    e_{s\bar{s}}(\bm{X}) + e_{\bar{s}\bar{s}}(\bm{X})  }.
\end{eqnarray*}

\begin{lemma}
\label{lemma::M-principal-score-zs}
Under Monotonicity, $\Pr(U=u  \mid Z=1,S=1,\bm{X}) = e_{1,u}(\bm{X})$ for $u=s\bar{s}$ and $ss$, and $\Pr(U=u\mid Z=0, S=0, \bm{X}) = e_{0,u}(\bm{X})$ for $u=s\bar{s}$ and $\bar{s}\bar{s}$.
\end{lemma}

\begin{proof}[Proof of Lemma \ref{lemma::M-principal-score-zs}]
Without essential loss of generality, we show only the case with $u=s\bar{s}$ under the treatment, and other cases are analogous.
We have
\begin{eqnarray*}
\Pr(U=s\bar{s}  \mid Z=1,S=1,\bm{X})
&=& \frac{ \Pr(Z=1,S=1, U=s\bar{s} \mid \bm{X})  }{  \Pr(Z=1,S=1\mid \bm{X}) } \\
&=&  \frac{ \Pr(Z=1, U=s\bar{s} \mid \bm{X})  }{  \Pr(Z=1, U=s\bar{s} \mid \bm{X}) + \Pr(Z=1, U=ss \mid \bm{X}) } .
\end{eqnarray*}
Randomization implies $Z\ind (U,\bm{X})$, and therefore
$$
\Pr(U=s\bar{s}  \mid Z=1,S=1,\bm{X})
= \frac{ \Pr(  U=s\bar{s} \mid \bm{X})  }{  \Pr(  U=s\bar{s} \mid \bm{X}) + \Pr(  U=ss \mid \bm{X}) }
= \frac{  e_{s\bar{s}}(\bm{X}) }{    e_{s\bar{s}}(\bm{X}) + e_{ss}(\bm{X}) } = e_{1, s\bar{s}}(\bm{X}) .
$$
\end{proof}

\begin{lemma}
\label{lemma::M-balance}
Under Monotonicity,
$\bm{X}\ind U \mid  \{ Z=1, S=1, e_{1, u}(\bm{X})\}$ for $u=s\bar{s}$ and $ss$, and
$\bm{X}\ind U \mid  \{ Z=0, S=0, e_{0, u}(\bm{X})\}$ for $u=s\bar{s}$ and $\bar{s}\bar{s}$.
\end{lemma}

\begin{proof}[Proof of Lemma \ref{lemma::M-balance}]
We prove only the case with $u=s\bar{s}$ under the treatment.
Conditional on $(Z=1,S=1)$, the principal strata can take only two values $s\bar{s}$ and $ss$, and the proof here follows similar arguments as the proof of Lemma \ref{lemma::SM-balance}.

First,
$
\Pr\{ U= s\bar{s} \mid  \bm{X}, Z=1,S=1, e_{1,s\bar{s}}(\bm{X}) \}  = \Pr( U= s\bar{s} \mid  \bm{X}, Z=1,S=1 ) = e_{1, s\bar{s}}(\bm{X}).
$
Second, by LIE we have
\begin{eqnarray*}
\Pr\{   U=s\bar{s} \mid Z=1, S=1, e_{1,s\bar{s}} (\bm{X}) \}  
&=& E[  \Pr\{   U=s\bar{s} \mid Z=1, S=1, e_{1,s\bar{s}} (\bm{X}) , \bm{X} \}   \mid Z=1, S=1, e_{1,s\bar{s}} (\bm{X})   ] \\
&=& E\{  e_{1,s\bar{s}} (\bm{X})\mid Z=1,S=1, e_{1,s\bar{s}} (\bm{X})   \} =  e_{1,s\bar{s}} (\bm{X}) .
\end{eqnarray*}
Therefore,
$
\Pr\{ U= s\bar{s} \mid  \bm{X}, Z=1,S=1, e_{1,s\bar{s}}(\bm{X}) \} =  \Pr\{   U=s\bar{s} \mid Z=1, S=1, e_{1,s\bar{s}} (\bm{X}) \},
$
and the conditional independence $\bm{X} \ind U \mid  \{ Z=1, S=1, e_{1, u}(\bm{X})\}$ follows.
\end{proof}

\begin{lemma}
\label{lemma::weaker-than-GPI}
Under Monotonicity, GPI implies that 
$$
Y(1)\ind U\mid (Z=1,S=1,\bm{X}),\quad Y(0)\ind U \mid  (Z=0,S=0, \bm{X}) .
$$
\end{lemma}

\begin{proof}[Proof of Lemma \ref{lemma::weaker-than-GPI}]
We prove only the first part; the second part is analogous.

GPI implies $Y(1)\ind U\mid \bm{X}$, and therefore $Y(1)\ind \{ U, \bm{1}_{ (U=s\bar{s} \text{ or } ss) } \} \mid  \bm{X}$.
Furthermore, we have $Y(1)\ind U\mid   \{ \bm{1}_{ (U=s\bar{s} \text{ or } ss) } , \bm{X} \}$. Because Randomization ensures that $Z$ is independent of all the potential outcomes and covariates, we have $Y(1)\ind U\mid   \{ Z, \bm{1}_{ (U=s\bar{s} \text{ or } ss) }, \bm{X} \}$, which further implies
$Y(1)\ind U\mid   \{ Z=1, \bm{1}_{ (U=s\bar{s} \text{ or } ss) }=1, \bm{X} \}$ or equivalently $Y(1)\ind U\mid  (Z=1,S=1,\bm{X})$.
\end{proof}

\begin{lemma}
\label{lemma::M-sufficient}
Under Monotonicity and GPI, we have
$Y(1)\ind U\mid \{  Z=1,S=1, e_{1,u}(\bm{X}) \}$ for $u=s\bar{s}$ and $ss$, and
$Y(0)\ind U\mid \{  Z=0,S=0, e_{0, u}(\bm{X}) \}$ for $u=s\bar{s}$ and $\bar{s}\bar{s}$.
\end{lemma}

\begin{proof}[Proof of Lemma \ref{lemma::M-sufficient}]
With Lemmas \ref{lemma::M-principal-score-zs}--\ref{lemma::weaker-than-GPI}, the proof follows from the same logic as that of Lemma \ref{lemma::SM-sufficiency}.
\end{proof}

\begin{proof}[Proof of Proposition 2]
It is straightforward to obtain
$$
E\{ Y(1)\mid U=\bar{s}\bar{s} \} = E(Y\mid Z=1,S=0), \quad
E\{ Y(0)\mid U=ss\} = E(Y\mid Z=0, S=1).
$$
Without loss of generality, we show only
$
E\{ Y(1)\mid U=s\bar{s} \} = E\{   w_{1,s\bar{s}} (\bm{X}) Y\mid Z=1,S=1\}.
$

First, by Randomization we have
$$
E\{ Y(1)\mid U=s\bar{s} \} =   E\{  Y(1) \mid Z=1, U=s\bar{s} \}  = E\{  Y(1)\mid Z=1, S=1, U=s\bar{s} \} .
$$
By LIE and Lemma \ref{lemma::M-sufficient}, we have
\begin{eqnarray*}
E\{ Y(1)\mid U=s\bar{s} \}
&=&
E[    E\{  Y(1)\mid Z=1, S=1, U=s\bar{s} , e_{1,s\bar{s}} (\bm{X}) \}  \mid Z=1,S=1,U=s\bar{s}    ]     \\
&=&
E[    E\{  Y(1)\mid Z=1, S=1 , e_{1,s\bar{s}} (\bm{X}) \}  \mid Z=1,S=1,U=s\bar{s}    ].
\end{eqnarray*}

Second, we have
$$
E\{   w_{1,s\bar{s}} (\bm{X}) Y\mid Z=1,S=1\}  =
E[  w_{1,s\bar{s}} (\bm{X})  E\{   Y(1)\mid Z=1,S=1, e_{1, s\bar{s}} (\bm{X}) \}    \mid Z=1,S=1 ] .
$$
Also we have
$$
f\{  e_{1, s\bar{s}} (\bm{X}) \mid U=s\bar{s}, Z=1, S=1 \}
= \frac{  f\{     e_{1, s\bar{s}} (\bm{X}) \mid Z=1, S=1 \}   \Pr\{  U=s\bar{s} \mid Z=1,S=1,  e_{1, s\bar{s}} (\bm{X}) \}   }{  \Pr(U=s\bar{s}  \mid Z=1,S=1 )  }.
$$
From the proof of Lemma \ref{lemma::M-balance}, we have $\Pr\{  U=s\bar{s} \mid Z=1,S=1,  e_{1, s\bar{s}} (\bm{X}) \}  = e_{1, s\bar{s}} (\bm{X})$. Simple algebra gives $\Pr(U=s\bar{s}  \mid Z=1,S=1 ) = \pi_{s\bar{s}} / (  \pi_{s\bar{s}} + \pi_{ss})$. Consequently, the weight satisfies
$$
 w_{1,s\bar{s}} (\bm{X}) = \frac{e_{1, s\bar{s}} (\bm{X})}{ \pi_{s\bar{s}} / (  \pi_{s\bar{s}} + \pi_{ss})}
 =\frac{ \Pr\{  U=s\bar{s} \mid Z=1,S=1,  e_{1, s\bar{s}} (\bm{X}) \}   }{  \Pr(U=s\bar{s}  \mid Z=1,S=1 )  }
= \frac{ f\{  e_{1, s\bar{s}} (\bm{X}) \mid U=s\bar{s}, Z=1, S=1 \}   }{ f\{     e_{1, s\bar{s}} (\bm{X}) \mid Z=1, S=1 \}  },
$$
which implies $E\{ Y(1)\mid U=s\bar{s} \} = E\{   w_{1,s\bar{s}} (\bm{X}) Y\mid Z=1,S=1\}$ according to Lemma \ref{lemma::importance-sampling}.
\end{proof}

\begin{proof}
[Proof of Corollary 2]
The theorem follows from Proposition 2 and zero PCEs on $h(\bm{X}).$
\end{proof}

In order to prove Proposition 4, we need an additional lemma.

\begin{lemma}
\label{lemma::M-sensitivity-lemma}
We can also represent the sensitivity parameters as
$$
\varepsilon_1 =   {  E\{ Y(1)\mid U=s\bar{s}, e_{1, u} (\bm{X}) \} \over  E\{  Y(1) \mid U=\bar{s}\bar{s}  , e_{1, u }(\bm{X}) \}  }
\quad (u=s\bar{s}, ss);
\qquad
\varepsilon_0 = { E\{  Y(0) \mid U= s\bar{s}, e_{0, u} (\bm{X}) \} \over  E\{  Y(0) \mid U=\bar{s}\bar{s} , e_{0, u}(\bm{X}) \}  }
\quad (u=s\bar{s}, \bar{s}\bar{s}).
$$
\end{lemma}

\begin{proof}[Proof of Lemma \ref{lemma::M-sensitivity-lemma}]
It follows from Lemma \ref{lemma::M-balance} and the method in the proof of Lemma \ref{lemma::SM-sensitivity-lemma}.
\end{proof}

\begin{proof}[Proof of Proposition 4]
We prove only that $E\{ Y(1)\mid U=s\bar{s} \} = E\{  w_{1,s\bar{s}}^{\varepsilon_1}(\bm{X}) Y \mid Z=1,S=1  \}$; other conditional expectations of the potential outcomes are analogous.

Randomization and LIE allow us to obtain that
\begin{eqnarray*}
E\{ Y(1)\mid U=s\bar{s} \} &=& E\{  Y(1)\mid Z=1, U=s\bar{s}, S=1\} \\
&=& E[   E\{  Y(1)\mid Z=1, U=s\bar{s}, S=1,  e_{1, s\bar{s}}(\bm{X}) \}  \mid   Z=1, U=s\bar{s}, S=1  ] .
\end{eqnarray*}
From the proofs of Lemmas \ref{lemma::M-balance}--\ref{lemma::M-sensitivity-lemma}, we also have
\begin{eqnarray*}
&&E \{ Y(1)\mid Z=1,S=1,e_{1,s\bar{s}} (\bm{X}) \} \\
&=& E\{  Y(1)\mid Z=1, U=s\bar{s}, S=1,  e_{1, s\bar{s}}(\bm{X}) \}     \Pr\{ U=s\bar{s} \mid Z=1,  S=1,  e_{1, s\bar{s}}(\bm{X})  \} \\
&& + E\{  Y(1)\mid Z=1, U=\bar{s}\bar{s}, S=1,  e_{1, s\bar{s}}(\bm{X}) \}  \Pr\{  U=\bar{s}\bar{s}\mid  Z=1, S=1,  e_{1, s\bar{s}}(\bm{X}) \} \\
&=& E\{  Y(1)\mid Z=1, U=s\bar{s}, S=1,  e_{1, s\bar{s}}(\bm{X}) \}   e_{1, s\bar{s}}(\bm{X})  + E\{  Y(1)\mid Z=1, U=\bar{s}\bar{s}, S=1,  e_{1, s\bar{s}}(\bm{X}) \}  e_{1, \bar{s}\bar{s}}(\bm{X}) \\
& =& \{   e_{1, s\bar{s}}(\bm{X})  + e_{1, \bar{s}\bar{s}}(\bm{X})/\varepsilon_1  \}  E\{  Y(1)\mid Z=1, U=s\bar{s}, S=1,  e_{1, s\bar{s}}(\bm{X}) \}  .
\end{eqnarray*}
Consequently, we obtain that
$$
E\{  Y(1)\mid Z=1, U=s\bar{s}, S=1,  e_{1, s\bar{s}}(\bm{X}) \} = \{   e_{1, s\bar{s}}(\bm{X})  + e_{1, \bar{s}\bar{s}}(\bm{X})/\varepsilon_1  \}^{-1}  E \{ Y(1)\mid Z=1,S=1,e_{1,s\bar{s}} (\bm{X}) \} ,
$$
implying that
\begin{eqnarray*}
E\{ Y(1)\mid U=s\bar{s} \}
= E\left[    \frac{\varepsilon_1}{\varepsilon_1  e_{1,s\bar{s}} (\bm{X}) + e_{1,\bar{s}\bar{s}} (\bm{X})}  E \{ Y(1)\mid Z=1,S=1,e_{1,s\bar{s}} (\bm{X}) \}  \mid   Z=1, U=s\bar{s}, S=1  \right].
\end{eqnarray*}

On the other hand, the weighted mean can be represented as
\begin{eqnarray*}
&&E\{  w_{1,s\bar{s}}^{\varepsilon_1}(\bm{X}) Y \mid Z=1,S=1  \}  \\
&=&   E [ w_{1,s\bar{s}} ^{\varepsilon_1} (\bm{X}) E\{   Y(1) \mid Z=1,S=1, e_{1,s\bar{s}}(\bm{X})  \} \mid  Z=1,S=1  ]  \\
&=&   E \left [   \frac{e_{1, s\bar{s}} (\bm{X})}{ \pi_{s\bar{s}} / (  \pi_{s\bar{s}} + \pi_{ss})}
    \frac{\varepsilon_1}{\varepsilon_1  e_{1,s\bar{s}} (\bm{X}) + e_{1,\bar{s}\bar{s}} (\bm{X})}
     E\{   Y(1) \mid Z=1,S=1, e_{1,s\bar{s}}(\bm{X})  \} \mid  Z=1,S=1  \right ].
\end{eqnarray*}
The proof of Proposition 2 shows that $  e_{1, s\bar{s}} (\bm{X}) / \{ \pi_{s\bar{s}} / (  \pi_{s\bar{s}} + \pi_{ss}) \}   $ is exactly the importance weight $  f\{  e_{1, s\bar{s}} (\bm{X}) \mid U=s\bar{s}, Z=1, S=1 \}    /  f\{     e_{1, s\bar{s}} (\bm{X}) \mid Z=1, S=1 \}  $, and the conclusion follows from Lemma \ref{lemma::importance-sampling} and the above expressions of $E\{ Y(1)\mid U=s\bar{s} \}$ and $E\{  w_{1,s\bar{s}}^{\varepsilon_1}(\bm{X}) Y \mid Z=1,S=1  \}.$
\end{proof}

\section{Proofs of the Results Without Monotonicity}
\label{sec::no-mono}

\begin{proof}[Proof of Proposition 5.]
The observed data impose the following restrictions: 
\begin{eqnarray*}
\left \{ \begin{array} {lll}
\pi_{ss} + \pi_{s\bar{s}} &=& p_1, \\
\pi_{ss} + \pi_{\bar{s}s} &=& p_0, \\
\pi_{ss} + \pi_{\bar{s}s} + \pi_{s\bar{s}} + \pi_{\bar{s}\bar{s}} &=& 1,\\
\pi_{\bar{s}s}  -  \xi \pi_{s\bar{s}} &=& 0.
\end{array}
\right.
\Longrightarrow
\left \{ \begin{array} {lll}
0\leq &\pi_{s\bar{s}}  = (p_1-p_0)/(1-\xi) &\leq 1 ,\\
0\leq &\pi_{ss}  =   p_1-(p_1-p_0)/(1-\xi) &\leq 1,\\
0\leq & \pi_{\bar{s}s} = \xi(p_1-p_0)/(1-\xi) &\leq 1,\\
0\leq & \pi_{\bar{s}\bar{s}}  = (1-p_0) - (p_1-p_0)/(1-\xi) &\leq 1,
\end{array}
\right.
\end{eqnarray*}
which further imply that $0\leq \xi \leq 1- ( p_1-p_0 )  / \min( p_1, 1-p_0) .$
\end{proof}

\begin{proof}[Proof of Proposition 6]
The same reasoning of the proof of Proposition 2 applies here.
\end{proof}

\section{EM Algorithms for Estimating Principal Scores}
\label{sec::principal-scores-computations}

\subsection{With Monotonicity}

Because $U$ takes three values, we can model $\Pr(U\mid\bm{X})$ as a three-level Multinomial Logistic model:
$$
\Pr(U = u\mid  \bm{X} ) = \frac{ \exp(\bm{\theta}_u^\top \bm{X}) }{  \sum_{u'} \exp(\bm{\theta}_{u'}^\top \bm{X})},
\quad (u=s\bar{s}, ss, \bar{s}\bar{s}) 
$$
where $\bm{\theta}_{s\bar{s}} = \bm{0}$ for identification.
Although we cannot fully observe $U$, we can 
and use the EM algorithm to find the MLEs by treating $U$ as missing data.

In the E-step, we first calculate the conditional probabilities of latent strata given the data $( Z_i=1,S_i=1,\bm{X}_i)$ and the parameters $(\bm{\theta}_{s\bar{s}}^{k}, \bm{\theta}_{\bar{s} \bar{s} } ^{k} ) $:
if $(Z_i=1, S_i=1)$, then
$$
\Pr(U_i=s\bar{s}\mid  - ) = \frac{1}{  1+ \exp(   \bm{\theta}_{ss}^{k\top} \bm{X}_i  ) }, \quad
\Pr(U_i=ss\mid -) = \frac{1}{  1+ \exp( -  \bm{\theta}_{ss}^{k\top} \bm{X}_i  ) } ;
$$
if $(Z_i=1, S_i=1)$, then $U_i=\bar{s}\bar{s}$;
if $(Z_i=0, S_i=1)$, then $U_i=ss$;
if $(Z_i=0, S_i=0)$, then
$$
\Pr(U_i = s\bar{s}\mid -) = \frac{1}{ 1+ \exp( \bm{\theta}_{\bar{s}\bar{s}}^{k\top} \bm{X}_i  ) }, \quad
\Pr( U_i =  \bar{s}\bar{s} \mid -) =  \frac{1}{   1+ \exp( -  \bm{\theta}_{\bar{s}\bar{s}}^{k\top} \bm{X}_i  ) }.
$$
We then create weighted samples:
for each $i$ with $(Z_i=1,S_i=1)$, we create two observations $(U_i=s\bar{s}, \bm{X}_i)$ and $(U_i=ss, \bm{X}_i) $ with weights $w_i^k=\Pr(U_i=s\bar{s}\mid  - )$ and $w_i^k=\Pr(U_i=ss\mid  - )$, respectively;
for each $i$ with $(Z_i=1,S_i=0)$, we create one observation $(U_i=\bar{s}\bar{s}, \bm{X}_i)$ with weight $w_i^k=1$;
for each $i$ with $(Z_i=0,S_i=1)$, we create one observation $(U_i = ss, \bm{X}_i)$ with weight $w_i^k=1$;
for each $i$ with $(Z_i=0,S_i=0)$, we create two observations $(U_i = s\bar{s}, \bm{X}_i)$ and $(U_i=\bar{s}\bar{s}, \bm{X}_i)$ with weights $w_i^k = \Pr(U_i = s\bar{s}\mid -) $ and $w_i^k = \Pr( U_i =  ss \mid -)$, respectively.

In the M-step, we fit a Multinomial Logistic model $\Pr(U\mid \bm{X})$ based on the weight samples created above, and update the parameters to be $(\bm{\theta}_{s\bar{s}}^{k+1}, \bm{\theta}_{\bar{s} \bar{s} } ^{k+1} ) $.

\subsection{Without Monotonicity}

We define $V_i = U_i$ if $U_i=ss$ or $\bar{s}\bar{s}$, and $V_i = s\&\bar{s}$ if $U_i=s\bar{s}$ or $\bar{s}s$.
First, we model $\Pr(V\mid \bm{X})$ as a three-level Multinomial Logistic regression:
$$
\Pr(V=v  \mid  \bm{X} ) = \frac{ \exp(\bm{\theta}_v^\top \bm{X}) }{  \sum_{v'} \exp(\bm{\theta}_{v'}^\top \bm{X})},
\quad (v=s\&\bar{s}, ss, \bar{s}\bar{s})
$$
where $\bm{\theta}_{s\&\bar{s}} = \bm{0}$. Second, we partition the category of $V$, $s\&\bar{s}$, into two sub-categories of $U$, $s\bar{s}$ and $\bar{s}s$, with probabilities $\Pr(U=s\bar{s} \mid V=s\&\bar{s} , \bm{X}) = 1/(1+\xi)$ and $\Pr(U=\bar{s}s\mid V=s\&\bar{s}, \bm{X}) = \xi/(1+\xi)$.
We use the EM algorithm to find the MLEs by treating $U$ as missing data.

In the E-step, we first calculate the conditional probabilities of latent strata given the data $( Z_i,S_i,\bm{X}_i)$ and the parameters $(\bm{\theta}_{ss}^{k}, \bm{\theta}_{\bar{s} \bar{s} } ^{k} ) $:
if $(Z_i=1,S_i=1)$, then
$$
\Pr(U_i=s\bar{s}\mid  - ) = \frac{1}{ 1+  (1+\xi) \exp(   \bm{\theta}_{ss}^{k\top} \bm{X}_i  ) }, \quad
\Pr(U_i=ss\mid -) = \frac{1}{ 1+ \exp( -  \bm{\theta}_{ss}^{k\top} \bm{X}_i  ) / (1+\xi) }  ;
$$
if $(Z_i=1,S_i=0)$, then
$$
\Pr(U_i=\bar{s}\bar{s} \mid -) = \frac{1}{ 1 + \xi \exp( -  \bm{\theta}_{\bar{s}\bar{s}}^{k\top} \bm{X}_i )/(1+\xi) } , \quad
\Pr(U_i = \bar{s} s \mid -) = \frac{1}{ 1 + (1+\xi) \exp( \bm{\theta}_{\bar{s}\bar{s}}^{k\top} \bm{X}_i  )/\xi  } ;
$$
if $(Z_i=0, S_i=1)$, then
$$
\Pr(U_i = ss\mid -) = \frac{1}{1+\xi\exp( - \bm{\theta}_{ss}^{k\top} \bm{X}_i)/(1+\xi)}, \quad
\Pr(U_i=\bar{s}s\mid -) = \frac{1}{1+(1+\xi)\exp( \bm{\theta}_{ss}^{k\top} \bm{X}_i )/\xi};
$$
if $(Z_i=0,S_i=0)$, then
$$
\Pr(U_i = s\bar{s}\mid -) = \frac{1}{   1+ (1+\xi) \exp( \bm{\theta}_{\bar{s}\bar{s}}^{k\top} \bm{X}_i  ) }, \quad
\Pr( U_i = \bar{s}\bar{s}  \mid -) =  \frac{1}{ 1+ \exp( -  \bm{\theta}_{\bar{s}\bar{s}}^{k\top} \bm{X}_i  )/(1+\xi) }.
$$
We then create a set of weighted samples:
for each $i$ with $(Z_i=1,S_i=1)$, we create two observations $(V_i=s\&\bar{s}, \bm{X}_i)$ and $(V_i=ss, \bm{X}_i) $ with weights $w_i^k=\Pr(U_i=s\bar{s}\mid  - )$ and $w_i^k=\Pr(U_i=ss\mid  - )$, respectively;
for each $i$ with $(Z_i=1, S_i=0)$, we create two observations $(V_i=\bar{s}\bar{s}, \bm{X}_i)$ and $(V_i=s\&\bar{s}, \bm{X}_i)$ with weights $w_i^k=\Pr(U_i=\bar{s}\bar{s} \mid -) $ and $w_i^k = \Pr(U_i = \bar{s} s \mid -)$;
for each $i$ with $(Z_i=0,S_i=1)$, we create two observations $(V_i = ss, \bm{X}_i)$ and $(V_i=s\&\bar{s}, \bm{X}_i)$ with weight $w_i^k=\Pr(U_i = ss\mid -) $ and $\Pr(U_i=\bar{s}s\mid -)$;
for each $i$ with $(Z_i=0,S_i=0)$, we create two observations $(V_i = s\&\bar{s}, \bm{X}_i)$ and $(V_i=\bar{s}\bar{s}, \bm{X}_i)$ with weights $w_i^k = \Pr(U_i = s\bar{s}\mid -) $ and $w_i^k = \Pr( U_i =  ss \mid -)$, respectively.

In the M-step, we fit a Multinomial Logistic model $\Pr(V\mid \bm{X} ) $ based on the weight samples created above, and update the parameters to be
$(\bm{\theta}_{s\bar{s}}^{k+1}, \bm{\theta}_{\bar{s} \bar{s} } ^{k+1} ) $.

\section{Explicit Forms of the Estimators}
\label{sec::weighting-estimators}

We present explicit forms of moment estimators by weighting and model-assisted estimators via covariate adjustment in Section 5. Let $N_1$ and $N_0$ be the treatment and control sample sizes.

%
%
%
%
%
%
%

\subsection{With Strong Monotonicity}

With PI, by Proposition 1 the moment estimators for PCEs are
\begin{eqnarray*}
\widehat{ACE}_{s\bar{s}}  &=&
\frac{1}{n_{11}} \sum_{  \{ i: Z_i=1, S_i=1\}  } Y_i
-   \frac{1}{N_0} \sum_{  \{i:Z_i=0\} }  \widehat{w}_{ s\bar{s} }(\bm{X}_i)    Y_i ,\\
\widehat{ACE}_{\bar{s}\bar{s}}  &=&
\frac{1}{n_{10}} \sum_{  \{ i: Z_i=1, S_i=0\}  } Y_i
-   \frac{1}{N_0}  \sum_{  \{i:Z_i=0\} }  \widehat{w}_{ \bar{s}\bar{s} }(\bm{X}_i)  Y_i,
\end{eqnarray*}
and the model-assisted estimators for PCEs are
\begin{eqnarray*}
\widehat{ACE}_{s\bar{s}}^{\textrm{adj}}  &=&
\frac{1}{n_{11}} \sum_{  \{ i: Z_i=1, S_i=1\}  } (Y_i   - \bm{\beta}_{1,s\bar{s}}^\top \bm{X}_i  )
 -  \frac{1}{N_0} \sum_{  \{i:Z_i=0\} }  \widehat{w}_{ s\bar{s} }(\bm{X}_i)    (Y_i - \bm{\beta}_{0, s\bar{s}}^\top \bm{X} )   \\
&& +  \frac{1}{ n_{11} + N_0  } (\bm{\beta}_{1,s\bar{s}} - \bm{\beta}_{0, s\bar{s}} )^\top   \left(   \sum_{  \{ i: Z_i=1, S_i=1\}  }   \bm{X}_i +  \sum_{  \{i:Z_i=0\} } \widehat{w}_{ s\bar{s} }(\bm{X}_i)  \bm{X}_i \right),
\\ 
\widehat{ACE}_{\bar{s}\bar{s}}^{\textrm{adj}}  &=&
\frac{1}{n_{10}} \sum_{  \{ i: Z_i=1, S_i=0\}  } (Y_i   - \bm{\beta}_{1,\bar{s}\bar{s}}^\top \bm{X}_i  )
 -  \frac{1}{N_0} \sum_{  \{i:Z_i=0\} }  \widehat{w}_{ \bar{s}\bar{s} }(\bm{X}_i)    (Y_i - \bm{\beta}_{0, \bar{s}\bar{s}}^\top \bm{X} )   \\
&& +  \frac{1}{ n_{10} + N_0  } (\bm{\beta}_{1,\bar{s}\bar{s}} - \bm{\beta}_{0, \bar{s}\bar{s}} )^\top   \left(   \sum_{  \{ i: Z_i=1, S_i=0\}  }   \bm{X}_i +  \sum_{  \{i:Z_i=0\} } \widehat{w}_{ \bar{s}\bar{s} }(\bm{X}_i)  \bm{X}_i \right).
\end{eqnarray*}

We choose
$
\bm{\beta}_{1,s\bar{s}}
$
as the linear regression coefficients of $Y_i$ on $\bm{X}_i$ using samples with $(Z_i=1, S_i=1) $,
$
\bm{\beta}_{0,s\bar{s}}
$
as the weighted linear regression coefficients of $Y_i$ on $\bm{X}_i$ using samples with $Z_i=0$ and weights $  w_{s\bar{s}} (\bm{X}_i) $,
$
\bm{\beta}_{1,\bar{s}\bar{s}}
$
as the linear regression coefficients of $Y_i$ on $\bm{X}_i$ using samples with $( Z_i=1, S_i=0) $, and
$
\bm{\beta}_{0,\bar{s}\bar{s}}
$
as the weighted linear regression coefficients of $Y_i$ on $\bm{X}_i$ using samples with $Z_i=0$ and weights $  w_{\bar{s}\bar{s}} (\bm{X}_i) $.

Without PI, we need only to change the estimates of the weights for a fixed sensitivity parameter $\varepsilon$ as in Proposition 3, and obtain estimators of the same forms as above.

\subsection{With Monotonicity}

With GPI, by Proposition 2 the moment estimators for PCEs are
\begin{eqnarray*}
\widehat{ACE}_{s\bar{s}}  &=&
\frac{1}{n_{11}}   \sum_{  \{ i: Z_i=1, S_i=1\}  } \widehat{w}_{1,s\bar{s}} (\bm{X}_i ) Y_i
-  \frac{1}{n_{00}} \sum_{  \{ i: Z_i=0, S_i=0\}  } \widehat{w}_{0,s\bar{s}} (\bm{X}_ i ) Y_i, \\
\widehat{ACE}_{\bar{s}\bar{s}}  &=&
\frac{1}{n_{10}}  \sum_{  \{ i: Z_i=1, S_i=0\}  }  Y_i
- \frac{1}{n_{00}}  \sum_{  \{ i: Z_i=0, S_i=0\}  } \widehat{w}_{0,\bar{s}\bar{s}} (\bm{X}_ i ) Y_i , \\
\widehat{ACE}_{ss}  &=&
\frac{1}{n_{11}} \sum_{  \{ i: Z_i=1, S_i=1\}  } \widehat{w}_{1,ss} (\bm{X}_i ) Y_i
-  \frac{1}{n_{01}}  \sum_{  \{ i: Z_i=0, S_i=1\}  }  Y_i,
\end{eqnarray*}
and the model-assisted estimators for PCEs are
\begin{eqnarray*}
\widehat{ACE}_{s\bar{s}}^{\textrm{adj}}  &=&
\frac{1}{n_{11}} \sum_{  \{ i: Z_i=1, S_i=1\}  } \widehat{w}_{1, s\bar{s} } (\bm{X}_i ) (Y_i   - \bm{\beta}_{1,s\bar{s}}^\top \bm{X}_i  )
 -  \frac{1}{n_{00}} \sum_{  \{i:Z_i=0, S_i=0\} }  \widehat{w}_{0, s\bar{s} }(\bm{X}_i)    (Y_i - \bm{\beta}_{0, s\bar{s}}^\top \bm{X} )   \\
&& +  \frac{1}{ n_{11} + n_{00}  } (\bm{\beta}_{1,s\bar{s}} - \bm{\beta}_{0, s\bar{s}} )^\top   \left(   \sum_{  \{ i: Z_i=1, S_i=1\}  }   \widehat{w}_{1, s\bar{s} }(\bm{X}_i)\bm{X}_i +  \sum_{  \{i:Z_i=0, S_i=0\} } \widehat{w}_{0, s\bar{s} }(\bm{X}_i)  \bm{X}_i \right),
\\ 
\widehat{ACE}_{\bar{s}\bar{s}}^{\textrm{adj}}  &=&
\frac{1}{n_{10}} \sum_{  \{ i: Z_i=1, S_i=0\}  } (Y_i   - \bm{\beta}_{1,\bar{s}\bar{s}}^\top \bm{X}_i  )
 -  \frac{1}{n_{00}} \sum_{  \{i:Z_i=0, S_i=0\} }  \widehat{w}_{0, \bar{s}\bar{s} }(\bm{X}_i)    (Y_i - \bm{\beta}_{0, \bar{s}\bar{s}}^\top \bm{X} )   \\
&& +  \frac{1}{ n_{10} + n_{00}  } (\bm{\beta}_{1,\bar{s}\bar{s}} - \bm{\beta}_{0, \bar{s}\bar{s}} )^\top   \left(   \sum_{  \{ i: Z_i=1, S_i=0\}  }   \bm{X}_i +  \sum_{  \{i:Z_i=0, S_i=0\} } \widehat{w}_{0, \bar{s}\bar{s} }(\bm{X}_i)  \bm{X}_i \right), \\
\widehat{ACE}_{ss}^{\textrm{adj}}  &=&
\frac{1}{n_{11}} \sum_{  \{ i: Z_i=1, S_i=1\}  } \widehat{w}_{1, ss }(\bm{X}_i)(Y_i   - \bm{\beta}_{1, ss}^\top \bm{X}_i  )
 -  \frac{1}{n_{01}} \sum_{  \{i:Z_i=0, S_i=1\} }  (Y_i - \bm{\beta}_{0, ss}^\top \bm{X} )   \\
&& +  \frac{1}{ n_{10} + n_{01} } (\bm{\beta}_{1, ss} - \bm{\beta}_{0, ss} )^\top   \left(   \sum_{  \{ i: Z_i=1, S_i=1\}  } \widehat{w}_{1, ss }  (\bm{X}_i)\bm{X}_i +  \sum_{  \{i:Z_i=0, S_i=1\} } \bm{X}_i \right) .
\end{eqnarray*}

We choose
$
\bm{\beta}_{1,u}
$
for $u=s\bar{s}$ and $ss$ as the weighted linear regression coefficients of $Y_i$ on $\bm{X}_i$ using samples with $(Z_i=1,S_i=1)$ and weights $  w_{1,u} (\bm{X}_i) $,
$
\bm{\beta}_{1,\bar{s}\bar{s}}
$
as the linear regression coefficients of $Y_i$ on $\bm{X}_i$ using samples with $(Z_i=1,S_i=0)$,
$
\bm{\beta}_{0,ss}
$
as the linear regression coefficients of $Y_i$ on $\bm{X}_i$ using samples with $(Z_i=0,S_i=1)$, and
$
\bm{\beta}_{0,u}
$
for $u=s\bar{s}$ and $\bar{s}\bar{s}$ as the weighted linear regression coefficients of $Y_i$ on $\bm{X}_i$ using samples with $(Z_i=0,S_i=0)$ and weights $  w_{0,u} (\bm{X}_i) $.

Without GPI, we need only to change the estimates of the weights for fixed sensitivity parameters $(\varepsilon_1, \varepsilon_0).$

\subsection{Without Monotonicity}
By Proposition 6  the estimators for PCEs are
\begin{eqnarray*}
\widehat{ACE}_{s\bar{s}}  &=&
\frac{1}{n_{11}}   \sum_{  \{ i: Z_i=1, S_i=1\}  } \widehat{w}_{1,s\bar{s}} (\bm{X}_i ) Y_i
-  \frac{1}{n_{00}} \sum_{  \{ i: Z_i=0, S_i=0\}  } \widehat{w}_{0,s\bar{s}} (\bm{X}_ i ) Y_i,
\\
\widehat{ACE}_{\bar{s}\bar{s}}  &=&
\frac{1}{n_{10}}   \sum_{  \{ i: Z_i=1, S_i=0\}  } \widehat{w}_{1,\bar{s}\bar{s}} (\bm{X}_i ) Y_i
-  \frac{1}{n_{00}} \sum_{  \{ i: Z_i=0, S_i=0\}  } \widehat{w}_{0,\bar{s}\bar{s}} (\bm{X}_ i ) Y_i,\\
\widehat{ACE}_{ss}  &=&
\frac{1}{n_{11}}   \sum_{  \{ i: Z_i=1, S_i=1\}  } \widehat{w}_{1,ss} (\bm{X}_i ) Y_i
-  \frac{1}{n_{01}} \sum_{  \{ i: Z_i=0, S_i=1\}  } \widehat{w}_{0,ss} (\bm{X}_ i ) Y_i,
\\
\widehat{ACE}_{\bar{s}s}  &=&
\frac{1}{n_{10}}   \sum_{  \{ i: Z_i=1, S_i=0\}  } \widehat{w}_{1,\bar{s}s} (\bm{X}_i ) Y_i
-  \frac{1}{n_{01}} \sum_{  \{ i: Z_i=0, S_i=1\}  } \widehat{w}_{0,\bar{s}s} (\bm{X}_ i ) Y_i,
\end{eqnarray*}
and the model-assisted estimators for PCEs are
\begin{eqnarray*}
\widehat{ACE}_{s\bar{s}}^{\textrm{adj}}  &=&
\frac{1}{n_{11}} \sum_{  \{ i: Z_i=1, S_i=1\}  } \widehat{w}_{1, s\bar{s} } (\bm{X}_i ) (Y_i   - \bm{\beta}_{1,s\bar{s}}^\top \bm{X}_i  )
 -  \frac{1}{n_{00}} \sum_{  \{i:Z_i=0, S_i=0\} }  \widehat{w}_{0, s\bar{s} }(\bm{X}_i)    (Y_i - \bm{\beta}_{0, s\bar{s}}^\top \bm{X} )   \\
&& +  \frac{1}{ n_{11} + n_{00}  } (\bm{\beta}_{1,s\bar{s}} - \bm{\beta}_{0, s\bar{s}} )^\top   \left(   \sum_{  \{ i: Z_i=1, S_i=1\}  }   \widehat{w}_{1, s\bar{s} }(\bm{X}_i)\bm{X}_i +  \sum_{  \{i:Z_i=0, S_i=0\} } \widehat{w}_{0, s\bar{s} }(\bm{X}_i)  \bm{X}_i \right), \\
\widehat{ACE}_{\bar{s}\bar{s}}^{\textrm{adj}}  &=&
\frac{1}{n_{10}} \sum_{  \{ i: Z_i=1, S_i=0\}  } \widehat{w}_{1, \bar{s}\bar{s} } (\bm{X}_i ) (Y_i   - \bm{\beta}_{1,s\bar{s}}^\top \bm{X}_i  )
 -  \frac{1}{n_{00}} \sum_{  \{i:Z_i=0, S_i=0\} }  \widehat{w}_{0, \bar{s}\bar{s} }(\bm{X}_i)    (Y_i - \bm{\beta}_{0, s\bar{s}}^\top \bm{X} )   \\
&& +  \frac{1}{ n_{10} + n_{00}  } (\bm{\beta}_{1,\bar{s}\bar{s}} - \bm{\beta}_{0, \bar{s}\bar{s}} )^\top   \left(   \sum_{  \{ i: Z_i=1, S_i=0\}  }   \widehat{w}_{1, \bar{s}\bar{s} }(\bm{X}_i)\bm{X}_i +  \sum_{  \{i:Z_i=0, S_i=0\} } \widehat{w}_{0, \bar{s}\bar{s} }(\bm{X}_i)  \bm{X}_i \right), \\
\widehat{ACE}_{ss}^{\textrm{adj}}  &=&
\frac{1}{n_{11}} \sum_{  \{ i: Z_i=1, S_i=1\}  } \widehat{w}_{1, ss }  (\bm{X}_i )  (Y_i   - \bm{\beta}_{1,ss}^\top \bm{X}_i  )
 -  \frac{1}{n_{01}} \sum_{  \{i:Z_i=0, S_i=1\} }  \widehat{w}_{0, ss }(\bm{X}_i)    (Y_i - \bm{\beta}_{0, ss}^\top \bm{X} )   \\
&& +  \frac{1}{ n_{11} + n_{01}  } (\bm{\beta}_{1,ss} - \bm{\beta}_{0, ss} )^\top   \left(   \sum_{  \{ i: Z_i=1, S_i=1\}  }   \widehat{w}_{1, ss }(\bm{X}_i)\bm{X}_i +  \sum_{  \{i:Z_i=0, S_i=1\} } \widehat{w}_{0, ss }(\bm{X}_i)  \bm{X}_i \right), \\
\widehat{ACE}_{\bar{s}s}^{\textrm{adj}}  &=&
\frac{1}{n_{10}} \sum_{  \{ i: Z_i=1, S_i=0\}  } \widehat{w}_{1, \bar{s}s }  (\bm{X}_i )  (Y_i   - \bm{\beta}_{1,\bar{s}s}^\top \bm{X}_i  )
 -  \frac{1}{n_{01}} \sum_{  \{i:Z_i=0, S_i=1\} }  \widehat{w}_{0, \bar{s}s }(\bm{X}_i)    (Y_i - \bm{\beta}_{0, \bar{s}s}^\top \bm{X} )   \\
&& +  \frac{1}{ n_{10} + n_{01}  } (\bm{\beta}_{1,\bar{s}s} - \bm{\beta}_{0, \bar{s}s} )^\top   \left(   \sum_{  \{ i: Z_i=1, S_i=0\}  }   \widehat{w}_{1, \bar{s}s }(\bm{X}_i)\bm{X}_i +  \sum_{  \{i:Z_i=0, S_i=1\} } \widehat{w}_{0, \bar{s}s }(\bm{X}_i)  \bm{X}_i \right) .
\end{eqnarray*}
We choose
$
\bm{\beta}_{1,u}
$
for $u=s\bar{s}$ and $ss$ as the weighted linear regression coefficients of $Y_i$ on $\bm{X}_i$ using samples with $(Z_i=1,S_i=1)$ and weights $  w_{1,u} (\bm{X}_i) $,
$
\bm{\beta}_{1,u}
$
for $u= s\bar{s}$ and $ \bar{s}\bar{s}$ as the weighted linear regression coefficients of $Y_i$ on $\bm{X}_i$ using samples with $(Z_i=1,S_i=0)$ and weights $  w_{1,u} (\bm{X}_i) $,
$
\bm{\beta}_{0,u}
$
for $u=\bar{s}s$ and $ss$ as the weighted linear regression coefficients of $Y_i$ on $\bm{X}_i$ using samples with $(Z_i=0,S_i=1)$ and weights $  w_{0,u} (\bm{X}_i) $, and
$
\bm{\beta}_{0,u}
$
for $u=s\bar{s}$ and $\bar{s}\bar{s}$ as the weighted linear regression coefficients of $Y_i$ on $\bm{X}_i$ using samples with $(Z_i=0,S_i=0)$ and weights $  w_{0,u} (\bm{X}_i) $.

\section{Additional Simulations}\label{sec::additional simulation}

We compare the finite sample performances of our model-assisted estimator with Jo and Stuart (2009)'s linear model-based estimator. The later assumes Strong Monotonicity, and we conduct simulations accordingly.

Let the sample size be $500$. We generate $X_{i1}, \ldots, X_{i4} \iidsim N(0,1)$ and $X_{i5} \sim \mathrm{Bern}(1/2),$ and let $\bm X_i = (1, X_{i1}, \ldots, X_{i5})^\top .$ We consider five cases indexed by the parameter $\theta=-1,-0.5,0,0.5,1$. For each $\theta,$ we generate principal strata by  
$
\mathrm{logit} ~ \mathrm{Pr}(U_i = s\bar{s}\mid \bm X_i)=\bm\theta^\top\bm X_i,
$
where $\bm{\theta}=(0,0.5,0.5,1,1,\theta)^\top$. We generate potential outcomes using a linear model as in Jo and Stuart (2009) by
$
Y_i(1) \mid \bm X_i \sim  N\left(\sum_{j=1}^5 X_{ij} + 2\cdot I_{\{U=s\bar{s}\}}+1, 1\right)
$ 
and
$
Y_i(0) \mid \bm X_i \sim  N\left(\sum_{j=1}^5 X_{ij} + 2, 1\right);
$
using a quadratic model by
$
Y_i(1) \mid \bm X_i \sim  N\left(\sum_{j=1}^5 X_{ij} + \sum_{j=1}^4 X^2_{ij} + 2\cdot I_{\{U=s\bar{s}\}}+1, 1\right)
$
and
$
Y_i(0) \mid \bm X_i \sim  N\left(\sum_{j=1}^5 X_{ij} + \sum_{j=1}^4 X^2_{ij} + 2, 1\right).
$

To examine the performance of the estimators with and without PI, in each scenario we analyze the data with and without the binary covariate $X_{i5},$ and respectively label the results as ``oracle'' and ``obs.'' Without using $X_{i5}$, $\theta$ is a measure of the violation from PI. Figure \ref{fg:simu} presents only the results for $ACE_{s \bar{s}}$, and omits similar results for other principal strata. We report the biases and standard errors of the estimators over $1000$ repeated samplings.

For linear case which favors the estimator in Jo and Stuart (2009), with the binary covariate both estimators has small biases, and without the binary covariate both estimators have similar bias issues. However, in both scenarios our estimator has smaller standard error. For the quadratic case in which the estimator in Jo and Stuart (2009) is mis-specified, with or without the binary covariate the estimator in Jo and Stuart (2009) suffers from severe bias issues, but our estimator has small biases. Again, in both scenarios our estimator has smaller standard error.

\begin{figure}[htbp]
\centering
\includegraphics[height=1.2\linewidth, width=0.8\linewidth]{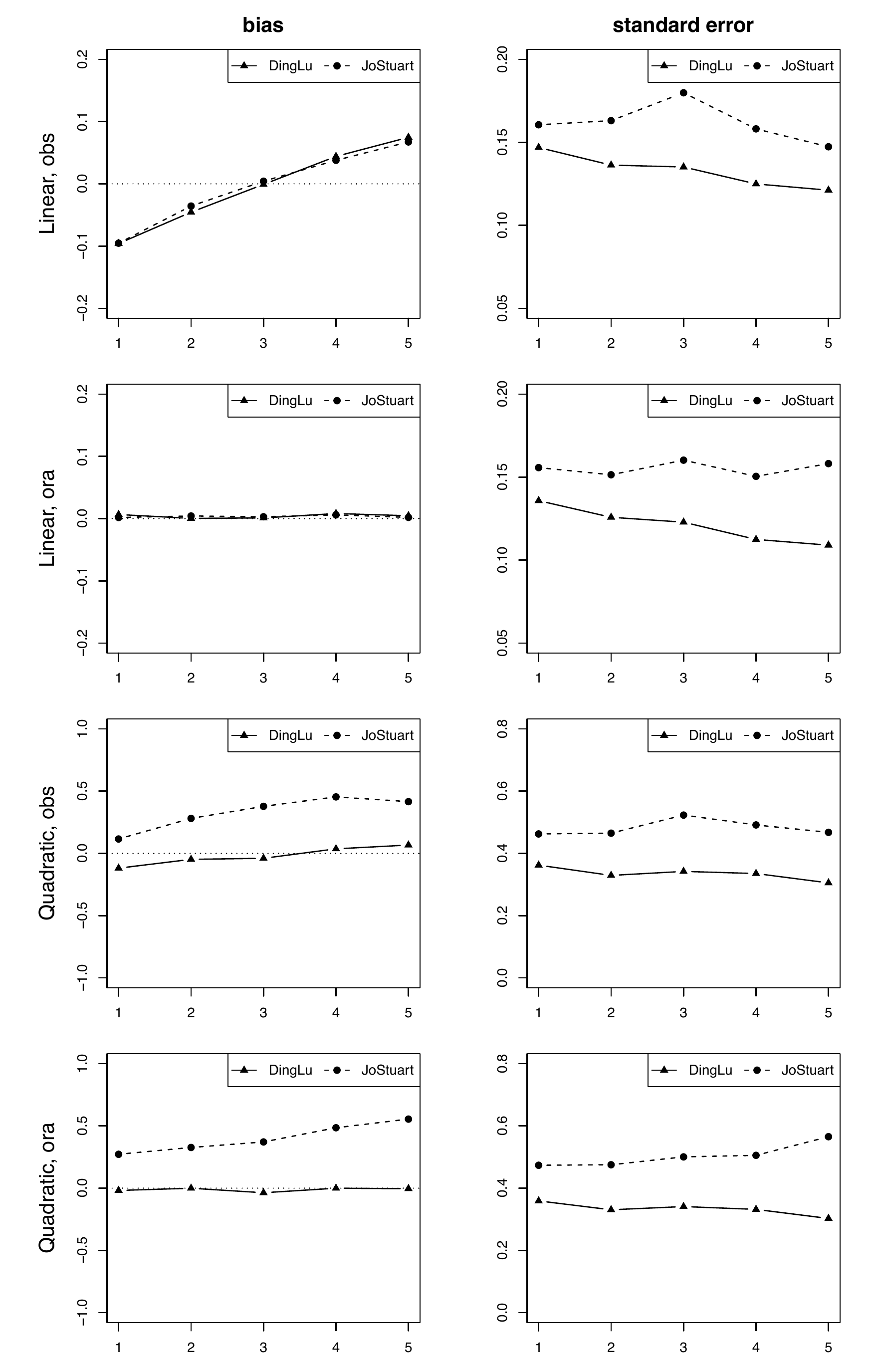}
\caption{
Comparison between our estimator (``DingLu'') and Jo and Stuart (2009)'s estimator (``JoStuart'') in terms of bias and standard error. The first two rows have linear outcome models, the last two rows have quadratic outcome models. The labels ``oracle'' and ``obs'' correspond to analysis with and without the binary covariate that ensures PI with other covariates. All of them are under Strong Monotonicity. The horizontal axis shows the case numbers.}
\label{fg:simu}
\end{figure}

\end{document}